\documentclass[pdflatex,sn-mathphys-num]{sn-jnl}

\usepackage{graphicx}%
\usepackage{multirow}%
\usepackage{amsmath,amssymb,amsfonts}%
\usepackage{amsthm}%
\usepackage{mathrsfs}%
\usepackage[title]{appendix}%
\usepackage{xcolor}%
\usepackage{textcomp}%
\usepackage{graphicx}
\usepackage{manyfoot}%
\usepackage{booktabs}%
\usepackage{algorithm}%
\usepackage{algorithmicx}%
\usepackage{algpseudocode}%
\usepackage{listings}%
\usepackage{dsfont}%
\usepackage[width=0.8\textwidth]{caption}

\theoremstyle{thmstyleone}%
\newtheorem{theorem}{Theorem}
\newtheorem{proposition}[theorem]{Proposition}%

\theoremstyle{thmstyletwo}%
\newtheorem{remark}{Remark}%
\newtheorem*{remark*}{Remark.}
\newtheorem*{notation*}{Notation.}
\theoremstyle{thmstylethree}%

\usepackage{caption}
\usepackage{subcaption}
\usepackage{graphicx}
\expandafter\def\csname ver@subfig.sty\endcsname{}
\usepackage[textwidth=8em,textsize=small]{todonotes}
\usepackage{amsmath}
\usepackage{natbib}
\usepackage{booktabs}   
\usepackage{array}      
\usepackage{tabularx}   
\usepackage{siunitx}    

\usepackage{float} 
\usepackage{placeins}

\newcommand{\Pro}{\mathbb{P}}

\usepackage{xcolor}
\definecolor{dgreen}{rgb}{0.,0.6,0.}

\usepackage{graphicx}
\usepackage{tikz}
\usetikzlibrary{bayesnet,positioning,fit,calc} 
\usetikzlibrary{arrows.meta,positioning}

\begin{document}

\title[Embedding Birth-Death Processes within a Dynamic Stochastic Block Model]{Embedding Birth-Death Processes within a Dynamic Stochastic Block Model}


\author*[1]{ \fnm{Gabriela Bayolo Soler}}\email{gbayolo@gmail.com}

\author[1]{\fnm{Miraine D\'avila Felipe}}

\author[1]{\fnm{Ghislaine Gayraud}}

\affil*[1]{Universit\'e de technologie de Compi\`egne, Alliance Sorbonne Universit\'e, LMAC, Compi\`egne, France}




\abstract{
Statistical clustering in dynamic networks aims to identify groups of nodes with similar or distinct internal connectivity patterns as the network evolves over time. While early research primarily focused on static Stochastic Block Models (SBMs), recent advancements have extended these models to handle dynamic and weighted networks, allowing for a more accurate representation of temporal variations in structure. Additional developments have introduced methods for detecting structural changes, such as shifts in community membership. However, limited attention has been paid to dynamic networks with variable population sizes, where nodes may enter or exit the network. To address this gap, we propose an extension of dynamic SBMs (dSBMs) that incorporates a birth-death process, enabling the statistical clustering of nodes in dynamic networks with evolving population sizes. This work makes three main contributions: (1) the introduction of a novel model for dSBMs with birth-death processes, (2) a framework for parameter inference and prediction of latent communities in this model, and (3) the development of an adapted Variational Expectation-Maximization (VEM) algorithm for efficient inference within this extended framework.
}


\maketitle

\section{Introduction}\label{sec1}

Network models play a crucial role in capturing and understanding population dynamics. Unlike models based on simple random interactions, real-world networks often exhibit complex structures that reflect underlying patterns and regularities. In the literature, there are many examples of stochastic network models that help to describe these structures, such as the random graph model introduced by Erd\H{o}s and R\'enyi \cite{erdHos1960evolution}, the small-world network presented by \cite{watts1998collective}, and the Stochastic Block Model (SBM) by \cite{holland1983stochastic}. These models have been applied across a broad spectrum of disciplines, including the social sciences, where they help explaining human behavior and social interactions; statistics, where they provide frameworks for analyzing relational data; physics, where they model phenomena such as percolation and diffusion; and biology,  where they help making predictions about the structure and dynamics of interbreeding populations.

In this work we focus on the statistical detection of groups in dynamic networks. The problem of identifying clusters, subsets of nodes that are more strongly or just differently connected internally than with the rest of the network, has attracted considerable attention in the statistics and network science literatures in recent years. Early work on statistical community detection largely concentrated on static SBMs, which model networks with a fixed vertex set and time-invariant connectivity (\cite{snijders1997estimation,holland1983stochastic,aicher2014learning}). While these models provide a powerful tool for uncovering latent structure in static networks, they do not account for the temporal evolution that is often observed in real-world systems.

This limitation has motivated a large body of work on dynamic extensions of SBMs and related models. 
Recent developments have extended static SBMs to weighted networks and to dynamic SBMs in which the edges evolve over time, either in discrete time (see, for example, \cite{matias2017statistical,tang2011dynamic,xu2014dynamic,jiang2023autoregressive}) or in continuous time (see \cite{ludkin2018dynamic,corneli2017dynamic}). These extensions allow one to study how relationships and community structures change over time, providing a more realistic description of many applications in which interactions are inherently dynamic. In parallel, there has been growing interest in detecting changes in community structure, including changepoint detection methods that identify abrupt shifts in network organization (\cite{peel2015detecting,wang2017change}) and approaches designed to capture more gradual, diffuse changes in community memberships or connectivity patterns (\cite{matias2017statistical,ludkin2018dynamic,xu2014dynamic}). Closely related block models have also been developed for bipartite or matrix data through the Latent Block Model \cite{brault2015review}.

A common limitation of most dynamic SBMs, however, is the \emph{fixed population} assumption: the node set $V$ is taken to be constant, whereas real systems often exhibit entry and exit through birth, death, migration or isolation. Some works allow transient presence/absence or introduce a special “inactive’’ class within a fixed population \cite{yang2011detecting,pmlr-v38-xu15,matias2017statistical,pensky2019dynamic}, and there is growing interest in regimes with structural shifts and phase transitions \cite{branzei2023phase}. Nevertheless, to the best of our knowledge, there remains relatively little methodology that \emph{jointly} models network evolution and genuinely \emph{evolving-population} dynamics. This motivates the birth-death SBM developed below, in which community sizes follow explicit demographic (birth-death) processes while edges are generated according to an SBM-type mechanism.

In this work, we address this gap by extending dynamic stochastic block models (dSBMs) to account for changes in population size through the integration of a birth-death process. Birth-death processes provide a natural probabilistic description of population dynamics, representing the arrivals (births) and departures (deaths) of individuals over time. By combining a dSBM for the network structure with a birth-death mechanism for the vertex set, we obtain a model that captures both the evolution of edges and the stochastic variation in population size.

More precisely, we introduce a \emph{dynamic stochastic block model with birth-death population dynamics}, which we refer to as the \emph{birth-death stochastic block model} (BD-SBM). In the BD-SBM, each node belongs to one of $K$ latent communities. When a new node is born, it inherits the community of its parent and retains this community membership throughout its lifetime, so that there is no community switching over time. At each observation time, conditional on the current community memberships, the edges follow a stochastic block model. Between observation times, the population size evolves according to a continuous-time birth-death process: new individuals enter the network and existing individuals leave, while the latent community structure that drives edge formation is preserved for all living individuals. This explicit coupling between population dynamics and network structure yields a richer and more realistic description of temporal networks in which both edges and vertices evolve.

Beyond its methodological interest, the BD-SBM is particularly suited to applications where both the population and the interaction structure evolve over time, and where community memberships are naturally inherited and stable. Examples include populations structured into families, lineages or clans in ecology and population genetics, where descendants belong to the community of their parents and remain in that group until death.
Social and organizational settings in which newcomers join pre-existing groups (such as departments, teams or schools) and do not change affiliation over the observation window. For example, academic research laboratories organized into thematic teams (e.g.\ probability, statistics, optimization, machine learning), where new PhD students or postdoctoral fellows typically join the team of their advisor and leave the system when they graduate, move to another institution or retire. In such contexts, the BD-SBM provides a coherent framework to jointly analyse the evolution of community structure and population size, while explicitly accounting for demographic turnover. In this context, the BD-SBM can be used to identify the teams of collaborations and the ways of collaborations of the researches.

From an inferential viewpoint, the BD-SBM raises non-trivial challenges. The presence of latent community labels and unobserved community sizes makes direct maximum likelihood (ML) estimation infeasible, as the observed-data likelihood would require summing over an enormous number of possible label configurations. To overcome this difficulty, we build on the mean–field variational framework introduced for static SBMs in \cite{daudin2008mixture} and later extended to dynamic SBMs with fixed population size in \cite{matias2017statistical}. We propose an adapted Variational Expectation–Maximization (VEM) algorithm tailored to the BD-SBM. Our variational family combines a mean–field approximation for individual labels with a structured approximation for the latent community sizes, encoded through additional variational parameters that describe the evolution of community sizes over time.

In summary, our main contributions are threefold:
\begin{enumerate}
    \item We introduce the BD-SBM, a dynamic stochastic block model in which the
          vertex set evolves according to a birth–death process, allowing the
          population size to vary over time.
    \item We address the challenging problem of parameter estimation and community
          detection in this setting, and develop a VEM algorithm adapted to the
          BD-SBM that jointly estimates the model parameters and the latent
          community memberships. In particular, we propose a structured variational
          family that augments the classical mean–field approximation with latent
          variables describing community sizes, thereby enforcing coherence between
          population dynamics and community structure.
    \item We assess the performance of the BD-SBM on both simulated temporal networks
          and a real-world collaboration network constructed from \texttt{arXiv} data.
\end{enumerate}

These developments extend the scope of stochastic block modelling to temporal networks with evolving population size, and provide a principled framework for studying the joint dynamics of communities and populations in a wide range of applications.

\section{The birth-death stochastic block model (BD-SBM)}
In this section we formally define the BD-SBM, a dynamic stochastic block model in which the population size evolves according to a birth-death process.

\subsection{Model description}

We consider a dynamic undirected network $G(t) = (V(t), E(t))$, indexed by time $t \geq 0$, which represents the interactions between individuals at $t$. The time $t_0$ denotes the initial observation time. For each $t$, $V(t)$ is the set of active vertices (individuals), and $E(t)$ is the set of edges (interactions) between any two vertices in $V(t)$. 
Since we consider only binary (undirected) edges, we use a notation shorthand: we denote indistinctly by $E(t)$ the set of edges and the adjacency matrix of the graph at $t$, that is, the matrix whose coefficients are $e_{ij}(t), \ 1\le i,j\le V(t)$, where $e_{ij}(t)$ is defined as the indicator of the presence of an edge between $i$ and $j$ at time $t$.
Lastly, since the number of individuals  varies over time, we write $N(t) = |V(t)|$ for the cardinality of $V(t)$.

We assume that once a node becomes active in the network, it is assigned to exactly one of the $K$ distinct communities $\{1, 2, \ldots,K\}$ to which it belongs throughout its lifetime. It means that once active, a node keeps its community membership during all its lifetime.
The number of communities $K$ is fixed and known in advance. We associate  to any individual $i$  its  community membership random vector $Z_i$ defined by   $Z_i = (Z_{ik})_{k=1,\ldots,K}$, where for all $k$, $Z_{ik} = \mathds{1}\left\{ i \ \text{is in community} \ k \right\}$ and $\sum_{k=1}^K Z_{ik} = 1$.

Randomness affects the temporal edges, the population size over time and the community memberships of individuals.

\begin{enumerate}
\item \textbf{Population dynamics and community membership.}
The BD-SBM assumes a finite population of $N$ nodes partitioned into $K$ disjoint
communities as follows:
\begin{enumerate}
    \item For each node $i \in V_0$, the community label
    $Z_i = (Z_{ik})_{k=1,\ldots,K}$ is drawn independently with
    \[
        \mathbb{P}(Z_{ik} = 1) = \beta_k, \qquad k = 1,\ldots,K,
    \]
    where $\beta = (\beta_k)_{k=1,\ldots,K}$ denotes the vector of community
    probabilities at time $t_0$.

    \item For each node $i \notin V_0$, the population dynamics are governed by a
    linear birth–death process with common birth rate $\lambda \geq 0$ and death
    rate $\mu \geq 0$, which are shared across communities. When a birth occurs at
    time $t$ and a new node $i$ is created, it inherits the community of its
    parent. Conditionally on a birth at time $t$, we have
    \[
        \mathbb{P}(Z_{ik} = 1)
        = \frac{\displaystyle\sum_{j \in V(t^-)} Z_{jk}}
               { N(t^-)}, \qquad k = 1,\ldots,K,
    \]
    where $t^-$ denotes an instant just before $t$ and $N(t^-) = |V(t^-)|$ is the
    population size at that time.

    Labels are assigned sequentially: the label of a newly created node is defined
    as one unit larger than the label of the previously created node, or, initially,
    larger than the number of nodes present at $t_0$. Once a node dies, its label
    no longer appears in subsequent vertex sets.
\end{enumerate}

\begin{remark*}
The assumption of common birth and death rates across communities allows us to derive
explicit maximum likelihood estimators for $(\lambda,\mu)$, as shown in
Subsection~\ref{sec:b-d_inf}. The more general case in which the birth–death
parameters are community-specific is discussed in
Appendix~\ref{app:class_specific_rates} as an extension of our VEM approach.
\end{remark*}

\item \textbf{Temporal edges.}
For $t \geq t_0$, given the partition of the vertices $V(t)$ into $K$ communities, the
graph $G(t)$ is generated according to a stochastic block model with connectivity
matrix $\pi = (\pi_{km})_{1 \leq k,m \leq K}$. More precisely, for all $i<j$ such
that $i,j \in V(t)$,
\[
    e_{ij}(t) \,\big|\, \{Z_{ik} Z_{jm} = 1\}
    \stackrel{\mathrm{ind.}}{\sim} \phi( \cdot , \pi_{km}),
\]
where $e_{ij}(t)$ denotes the edge between $i$ and $j$ at time $t$, and
$\phi( \cdot , \pi_{km})$ stands for the probability mass function of the Bernoulli distribution with success probability
$\pi_{km}$.

\begin{notation*} 
In the interest of clarity, we avoid the introduction of an extra notation to distinguish between the random variables $e_{ij}(t)$ and their realisations. Whenever this notation is used, its role is whether specified or can be easily deduced from the context.
\end{notation*}
\end{enumerate}

We refer to Figure~\ref{fig:bdsbm_schematic_orthogonal_vhd} for a hierarchical overview of our BD-SBM.

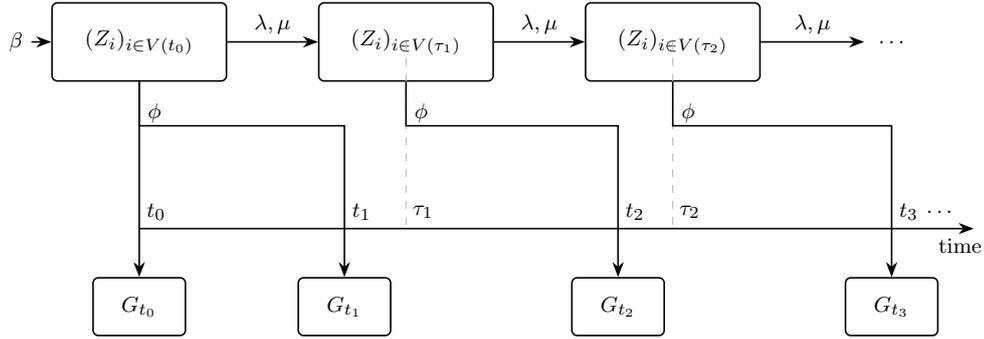
\begin{figure}[h]
\captionsetup{
      width=\textwidth,   
      margin=0.99cm       
    }
    \centering
\begin{tikzpicture}[
 scale=0.9,
  every node/.style={transform shape}, 
  >=Stealth,
  font=\small,
  lab/.style={font=\small},
  big/.style={
    draw, line width=0.6pt, rounded corners=2pt,
    minimum width=2.55cm, minimum height=1.15cm,
    align=center, inner sep=3pt
  },
  small/.style={
    draw, line width=0.6pt, rounded corners=2pt,
    minimum width=1.35cm, minimum height=0.85cm,
    align=center, inner sep=2pt
  },
  arr/.style={->, line width=0.6pt},
  arrlab/.style={midway, right, yshift=-4pt},
  vguide/.style={draw=gray!55, dashed, line width=0.4pt}
]

\node[big] (V0) at (0,0.65) {$(Z_i)_{i \in V(t_0)}$};

\node[big] (V1) at (3.9,0.65) {$(Z_i)_{i \in V(\tau_1)}$};

\node[big] (V2) at (7.8,0.65) {$(Z_i)_{i \in V(\tau_2)}$};

\node at (11,0.65) {$\ldots$};

\node (beta) at (-1.8,0.65) {$\beta$};
\draw[arr] (beta) -- (V0.west);

\draw[arr] (V0.east) -- node[above, font=\small] {$\lambda,\mu$} (V1.west);
\draw[arr] (V1.east) -- node[above, font=\small] {$\lambda,\mu$} (V2.west);
\draw[arr] (V2.east) -- node[above, font=\small] {$\lambda,\mu$} (10.6,0.65);

\def\yaxis{-2.10}
\draw[arr] (0.0,\yaxis) -- (12.2,\yaxis) node[right] {};
\node[lab] at (0.25,\yaxis+0.25) {$t_0$};
\node[lab] at (3.25,\yaxis+0.25) {$t_1$};
\node[lab] at (4.15,\yaxis+0.25) {$\tau_1$};
\node[lab] at (7.25,\yaxis+0.25) {$t_2$};
\node[lab] at (8.05,\yaxis+0.25) {$\tau_2$};
\node[lab] at (11.25,\yaxis+0.25) {$t_3$};
\node[lab] at (11.7,\yaxis+0.25) {$\ldots$};
\node[lab] at (12,\yaxis-0.25) {time};

\draw[vguide] (3.9,0.65) -- (3.9,\yaxis); 
\draw[vguide] (7.8,0.65) -- (7.8,\yaxis); 

\def\ysnap{-3.25}
\node[small] (E0) at (0.0,\ysnap) {$G_{t_0}$};
\node[small] (E1) at (3.0,\ysnap) {$G_{t_1}$};
\node[small] (E3) at (7.0,\ysnap) {$G_{t_2}$};
\node[small] (E4) at (11.0,\ysnap) {$G_{t_3}$};

\draw[arr]
  (V0.south) -- ++(0,-0.65) node[arrlab] {} -| (E0.north);

\draw[arr]
  (V0.south) -- ++(0,-0.65) node[arrlab] {$\phi$} -| (E1.north);

\draw[arr]
  (V1.south) -- ++(0,-0.65) node[arrlab] {$\phi$} -| (E3.north);

\draw[arr]
  (V2.south) -- ++(0,-0.65) node[arrlab] {$\phi$} -| (E4.north);

\end{tikzpicture}

\caption{Graphical representation of the BD--SBM model: community labels \(Z_i\) are defined on the alive set \(V(t)\), which evolves in continuous time under a birth--death process with rates \(\lambda\) and \(\mu\). The birth or death events occurs at ordered times $\tau = \left\{ \tau_1, \tau_2, \dots \right\}$. At observation times $t_\ell$, the snapshot graph \(G_{t_\ell}\) is generated from the current labels via the edge mechanism \(\phi\).}
\label{fig:bdsbm_schematic_orthogonal_vhd}
\end{figure}

Let us define $\theta$ as the model parameters 
$$\theta = (\lambda, \mu, \pi, \beta) \in \mathbb{R}_+\times  \mathbb{R}_+ \times [0,1]^{K^2} \times [0,1]^K.$$ 

\subsection{Observations and latent variables}\label{observations}

Snapshots of the dynamic undirected network $G(t) = (V(t), E(t))$ are observed at discrete times $t \in \mathcal{T} = \{t_0, t_1, \ldots, t_T\}$. We denote by $V_0$ the set of vertices at time $t_0$ and by $N_0 := \lvert V_0 \rvert$ its cardinality. The birth–death process is assumed to be observed continuously, meaning that the ordered sequence of birth–death event times, namely 
\[
\tau = \left\{ \tau_1, \tau_2, \dots, \tau_M\right\}, \qquad t_0\leq \tau_1 \text{ and } \tau_M \leq t_T,
\]
the nature of these events (birth or death), and the label of the individual undergoing each birth or death event are all observed. Consequently, the number of vertices $N(t)$ is known at any  $t \in [t_0, t_T]$.

The community affiliations of the individuals $i \in V := \bigcup_{t \in \tau } V(t)$ are unobserved, and hence ${Z}= \left\{ Z_i \right\}_{i\in V}$  is a collection of latent vectors. Let $N := \lvert V \rvert$ be the total number of individuals present in the system over the time interval $[t_0,t_T]$.

For clarity, we summarize below the observed data and the latent variables involved in the model:

\begin{enumerate}
\item \textbf{Observed variables:}
\begin{enumerate}
     \item the set of edges at times in $\mathcal{T} = \{t_0,\ldots,t_T\}$, the so-called \textit{snapshots}, defined as,
    \[
    e(\mathcal{T}) := \bigl\{(e_{ij}(t_\ell ))_{i,j}\bigr\}_{\ell =0,\ldots,T};
    \]
    \item  the ordered sequence of birth–death event times $\tau = \left\{ \tau_1, \tau_2, \dots, \tau_M\right\}$, with $t_0\leq \tau_1$  and  $\tau_M \leq t_T$; for notation convenience, we introduce an additional artificial event time $\tau_0$ which coincides with $t_0$, i.e., $\tau_0 := t_0$;
    \item the associated individuals and nature of events of the birth–death process, described by the sequence
    \[
    b(\tau) := \left\{(i_\ell, b_\ell)\right\}_{\ell=1,\ldots,M},
    \]
    where $i_\ell$ denotes the label of the individual undergoing the birth or death event at time $\tau_\ell$, and $b_\ell$ indicates whether a birth ($b_\ell = 1$) or a death ($b_\ell = -1$) occurs at time $\tau_\ell$. Consequently, the set $V(\tau_\ell)$, and the number $N(\tau_\ell)$, of alive individuals, is known for every $\ell = 0,\ldots,M$. \\
    Moreover, for each individual $i$, we denote by $\tau_i^b$ and $\tau_i^d$ her/his birth and death times, respectively, with the convention $\tau_i^b = t_0$ for $i$ present in $V_0$ and $\tau_i^d = t_T$ for $i$ present in $V(t_T)$. We also denote by $T_B$ and $T_D$ the set of birth times, respectively death times:
    $$ T_B = \{\tau_\ell : b_\ell = 1 \}, \qquad T_D = \{\tau_\ell : b_\ell = -1 \}. $$
    \end{enumerate}
    \item \textbf{Latent variables:}
    \begin{enumerate}
    \item the community membership random vectors ${Z} = \left\{ Z_i \right\}_{i \in V}$, where, for all $i \in V$,  $Z_i = (Z_{ik})_{k=1,\ldots,K}$ with  $Z_{ik} = \mathds{1} \left\{i \ \text{is in community} \ k \right\}$ and $\sum_{k=1}^K Z_{ik} = 1$.
    \end{enumerate}
\end{enumerate}

Figure~\ref{fig:bd_sbm_schematic} illustrates the generative mechanism of the BD-SBM, showing the evolution of active node set and community labels, and the associated observed snapshots.

\begin{figure}[h]
\captionsetup{
      width=\textwidth,   
      margin=0.99cm       
    }
\centering
\begin{tikzpicture}[scale=0.85, x=1cm,y=1cm, font=\small, >=Stealth]

\definecolor{commA}{RGB}{52,121,199}
\definecolor{commB}{RGB}{231,122,61}
\definecolor{commC}{RGB}{76,175,120}

\newcommand{\birthmark}[2]{\fill (#1,#2) circle (1.1pt);}
\newcommand{\deathmark}[2]{%
  \draw[line width=0.5pt] (#1-0.10,#2-0.10) -- (#1+0.10,#2+0.10);
  \draw[line width=0.5pt] (#1-0.10,#2+0.10) -- (#1+0.10,#2-0.10);
}

\newcommand{\AdjSnapshot}[7]{%
\begin{scope}[shift={(#1,#2)}]
  \def\s{#3}
  \def\nA{#5}\def\nB{#6}\def\nC{#7}
  \pgfmathtruncatemacro{\nTot}{\nA+\nB+\nC}
  \pgfmathsetmacro{\dx}{\s/\nTot}
  \pgfmathtruncatemacro{\bOne}{\nA}
  \pgfmathtruncatemacro{\bTwo}{\nA+\nB}
  \pgfmathtruncatemacro{\nGrid}{\nTot-1}

  \fill[gray!04] (0,0) rectangle (\s,-\s);

  \fill[commA!14] (0,0) rectangle (\bOne*\dx,-\bOne*\dx);
  \fill[commB!14] (\bOne*\dx,-\bOne*\dx) rectangle (\bTwo*\dx,-\bTwo*\dx);
  \fill[commC!14] (\bTwo*\dx,-\bTwo*\dx) rectangle (\s,-\s);

  \draw[line width=0.55pt] (0,0) rectangle (\s,-\s);
  \foreach \i in {1,...,\nGrid}{
    \draw[gray!22,line width=0.25pt] (\i*\dx,0) -- (\i*\dx,-\s);
    \draw[gray!22,line width=0.25pt] (0,-\i*\dx) -- (\s,-\i*\dx);
  }

  \draw[line width=0.50pt] (\bOne*\dx,0) -- (\bOne*\dx,-\s);
  \draw[line width=0.50pt] (\bTwo*\dx,0) -- (\bTwo*\dx,-\s);
  \draw[line width=0.50pt] (0,-\bOne*\dx) -- (\s,-\bOne*\dx);
  \draw[line width=0.50pt] (0,-\bTwo*\dx) -- (\s,-\bTwo*\dx);

  \def\celldot##1##2{%
    \fill[black!70] ({(##1-0.5)*\dx},{-(##2-0.5)*\dx}) circle (0.65pt);
  }
  \def\symdot##1##2{%
    \celldot{##1}{##2}\celldot{##2}{##1}
  }

  \ifnum\nA>1 \symdot{1}{2}\fi
  \ifnum\nB>1 \symdot{\numexpr\bOne+1\relax}{\numexpr\bOne+2\relax}\fi
  \ifnum\nC>1 \symdot{\numexpr\bTwo+1\relax}{\numexpr\bTwo+2\relax}\fi

  \ifnum\nA>0\ifnum\nB>0 \symdot{1}{\numexpr\bOne+1\relax}\fi\fi
  \ifnum\nB>0\ifnum\nC>0 \symdot{\numexpr\bOne+1\relax}{\numexpr\bTwo+1\relax}\fi\fi
  \ifnum\nA>0\ifnum\nC>0 \symdot{1}{\numexpr\bTwo+1\relax}\fi\fi

  \node[below=3pt] at (\s/2,-\s) {#4};
  \node[below=14pt, font=\scriptsize, gray!60] at (\s/2,-\s) {$|V(t)|=\nTot$};
\end{scope}
}

\def\xmin{1.6}
\def\xmax{10.6}
\def\yaxis{3.8}

\def\tA{2.0}
\def\tB{5.0}
\def\tC{8.0}

\def\ySnap{1.55}
\def\sSnap{2.2}
\def\halfSnap{1.1}

\pgfmathsetmacro{\xTimeStart}{\tA}      
\pgfmathsetmacro{\xTimeEnd}{\xmax+3.0}  
\draw[->, line width=0.8pt] (\xTimeStart,\yaxis) -- (\xTimeEnd,\yaxis) node[right] {time};

\foreach \t/\lab in {\tA/$t_0$,\tB/$t_1$,\tC/$t_2$}{
  \draw[dashed, gray!55, line width=0.35pt] (\t,\yaxis+3.05) -- (\t,\ySnap+0.10);
  \draw[line width=0.8pt] (\t,\yaxis-0.10) -- (\t,\yaxis+0.10);
  \node[below] at (\t,\yaxis-0.15) {\lab};
}

\node[anchor=west] at (\xmin,7.55)
{\textbf{Continuous-time population (birth--death) with labels $Z_i$}};

\foreach \b/\d/\y/\c in {
  2.0/8.1/7.0/commA,
  2.0/6.2/6.6/commB,
  2.0/10.0/6.2/commC,
  2.6/8.9/5.8/commA,
  3.0/4.9/5.4/commB,
  3.7/9.7/5.0/commC,
  4.2/10.0/4.6/commB,
  5.4/7.6/4.2/commB}{
  \draw[line width=1.9pt, draw=\c, line cap=round] (\b,\y) -- (\d,\y);
  \birthmark{\b}{\y}
  \deathmark{\d}{\y}
}

\begin{scope}[shift={(12,6.75)}]
  \filldraw[
    fill=white, fill opacity=0.92,
    draw=gray!60, draw opacity=1,
    line width=0.4pt, rounded corners=2pt
  ] (-0.25,0.55) rectangle (2.2,-0.55);

  \begin{scope}[shift={(-0.05,0.05)}]
    \draw[line width=1.9pt, commA] (0,0.25) -- (0.6,0.25);
    \node[anchor=west] at (0.7,0.25) {$Z_{i1}=1$};
    \draw[line width=1.9pt, commB] (0,-0.05) -- (0.6,-0.05);
    \node[anchor=west] at (0.7,-0.05) {$Z_{i2}=1$};
    \draw[line width=1.9pt, commC] (0,-0.35) -- (0.6,-0.35);
    \node[anchor=west] at (0.7,-0.35) {$Z_{i3}=1$};
  \end{scope}
\end{scope}

\node[anchor=west] at (\xmin,2.25)
{\textbf{Observed network snapshots (edges among currently alive nodes)}};

\AdjSnapshot{\tA-\halfSnap}{\ySnap}{\sSnap}{snapshot at $t_0$}{1}{1}{1}
\AdjSnapshot{\tB-\halfSnap}{\ySnap}{\sSnap}{snapshot at $t_1$}{2}{2}{2}
\AdjSnapshot{\tC-\halfSnap}{\ySnap}{\sSnap}{snapshot at $t_2$}{2}{1}{2}

\end{tikzpicture}
\caption{Schematic BD-SBM with $K=3$ communities: individuals appear and disappear following a continuous time birth--death process. For $i=1,\ldots, 8$, individual $i$ has a a community label defined as a vector $Z_i = (Z_{i})_{1\le i\le 3}$, where $Z_{ik}=1$ if individual $i$ is in community $k$, and 0 otherwise.  Communities are represented by different colors. At observation times $t_0,t_1,t_2$, we are informed by a network snapshot among the individuals alive at that time; which is represented by the adjacency matrix of the network at that time. The size of the population at snapshot time $t_m$ is $|V(t_m)|$. Diagonal blocks, corresponding to interactions from individuals in the same community, are highlighted. Interactions are symmetric and represented by dots centered in cells.}
\label{fig:bd_sbm_schematic}
\end{figure}
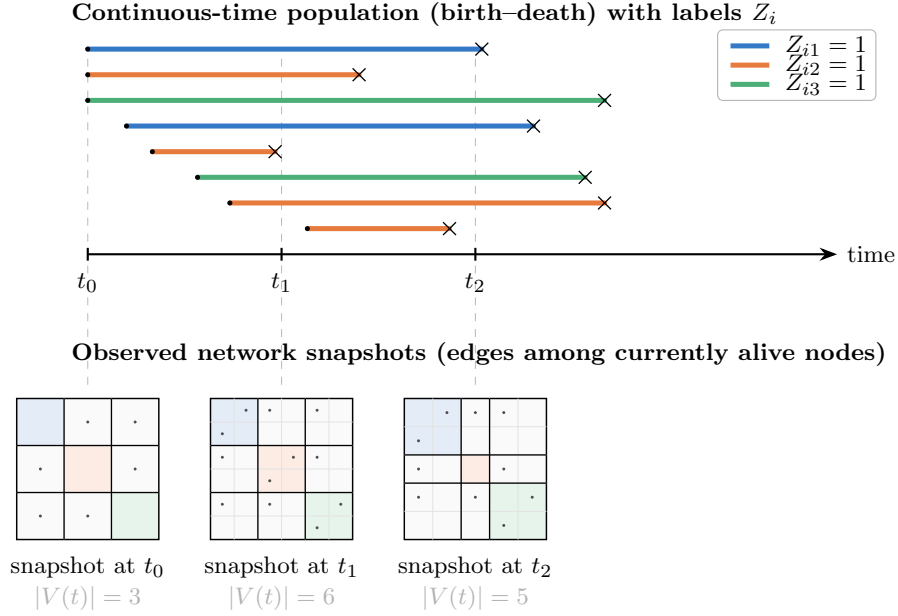

\medskip 

For convenience, we introduce the latent community-size indicators at the jump times of the birth-death process. 
For each event time $\tau_\ell$ with index $\ell = 0,\ldots,M$, each community $k = 1,\ldots,K$ and each integer $n \geq 0$, we define,
\[L^\ell = \left\{ L_k^\ell \right\}_k = (L_{k,n}^\ell)_{k,n}, \quad \text{where} \quad L_{k,n}^\ell = \mathds{1} \left\{  \sum_{i \in V(\tau_\ell)} Z_{ik} = n \right\}.\]

Thus, $L_{k,n}^{\ell} = 1$ if and only if community $k$ has size $n$ at time $\tau_\ell$. 
These variables satisfy
\begin{align}
    \sum_{n=0}^{N(\tau_\ell)} L_{k,n}^{\ell} &= 1, \nonumber \\
    \sum_{n=0}^{N(\tau_\ell)} n\,L_{k,n}^{\ell} &= \sum_{i\in V(\tau_\ell)} Z_{ik},\nonumber\\
    \sum_{k=1}^K \sum_{n=0}^{N(\tau_\ell)} n\,L_{k,n}^{\ell} &= N(\tau_\ell),\nonumber
\end{align}
where $N(\tau_\ell) = \lvert V(\tau_\ell)\rvert $ denotes the number of individuals alive at time $\tau_\ell$.

Moreover, the size of each community evolves over time according to a birth–death Markov chain and can change by at most one unit between two consecutive jump times. This is encoded, for every $\ell = 1,\ldots,M$, and every $k\in\{1,\ldots,K\}$, by
\begin{align}
\text{there is a birth at time } \tau_\ell 
&\quad \Leftrightarrow \quad 
\left\{  
\begin{array}{lcl}
\displaystyle \sum_{k,n} L_{k,n+1}^{\ell} L_{k,n}^{\ell-1} &=& 1,\\[0.2cm]
\displaystyle \sum_{k, n' \neq n+1} L_{k,n'}^{\ell} L_{k,n}^{\ell-1} &=& 0,
\end{array} 
\right. \nonumber  \\
\text{there is a death at time } \tau_\ell  
&\quad \Leftrightarrow \quad 
\left\{  
\begin{array}{lcl}
\displaystyle \sum_{k, n} L_{k,n-1}^{\ell} L_{k,n}^{\ell-1} &=& 1,\\[0.2cm]
\displaystyle \sum_{k, n' \neq n-1} L_{k,n'}^{\ell} L_{k,n}^{\ell-1} &=& 0.
\end{array} 
\right. \nonumber
\end{align}

\section{Inference for the BD-SBM}

 From the observed data, recall that our inferential goals are twofold: (i) to recover the classes to which each individual belongs (latent community membership random vector), and (ii) to estimate the community- and population-dynamic parameters, that is, the model parameter vector $\theta = (\lambda, \mu, \pi, \beta)$.

Models with latent variables, such as the one considered here, pose specific challenges for maximum likelihood estimation because the likelihood of the observed data involves integrating over all possible configurations of the latent variables. In our setting, the observed-data likelihood of $(\tau, e(\mathcal{T} ), b(\tau))$ is obtained by summing the complete-data likelihood over all possible assignments of the latent community memberships ${Z}$. Except in very small systems (e.g.\ $N$ and $M$ small), this summation is not tractable and the observed-data likelihood does not admit a closed-form expression. As a consequence, computing or maximizing the likelihood directly is computationally prohibitive, which would severely limit the practical use of the model.

To overcome this challenge, the inference is usually carried out by algorithms that either find or approximate   locally maxima of the observed-data log-likelihood.  A classical approach is the expectation-maximization (EM) algorithm, which alternates between computing the conditional expectation of the complete-data log-likelihood (E-step) and maximizing this expectation with respect to the parameters (M-step). In stochastic block models, however, the conditional distribution of the latent labels given the observed edges does not factorize across individuals, so the E-step is not available in closed form and becomes computationally intractable.

This motivates the use of variational approximations to EM (VEM; see \cite{jordan1999introduction}). For static SBMs, VEM was introduced by \cite{daudin2008mixture} and has since been extended in various directions; see, for instance, \cite{matias2014modeling} for an overview and comparisons with alternative estimators. In parallel, Markov chain Monte Carlo (MCMC) methods have also been proposed for dynamic network models, providing accurate posterior approximations at the price of a substantial computational cost; see for example \cite{ludkin2018dynamic,pmlr-v31-dubois13a}. In contrast, variational methods offer a scalable alternative by optimizing a tractable lower bound on the evidence (ELBO), and have been successfully applied to dynamic SBMs \cite{matias2017statistical,corneli2017dynamic}. In the setting of our BD-SBM,  we adopt a Variational Expectation-Maximization (VEM) strategy. 

While the birth–death parameters $(\lambda,\mu)$ admit closed-form maximum likelihood estimators,  the SBM parameters $(\beta,\pi)$  together with the class at which each individual belongs do not; hence,
we  use a variational approximation to estimate the latter. 

In the remainder of this section, we first derive the complete-data likelihood, then describe the estimation of the birth–death parameters, introduce the variational family, and finally present the VEM algorithm, model selection criterion, and initialization procedure.

\subsection{Complete-data likelihood}

We begin by deriving the complete-data likelihood of $\big({Z}, \tau, e(\mathcal{T} ), b(\tau)\big)$ parametrized by $\theta=(\lambda,\mu,\pi,\beta)$ and associated to the probability  distribution denoted $\Pro_{\theta}$. By a slight  abuse of notation, we denote this  complete-data likelihood 
by $\Pro_{\theta}\big({Z}, \tau, e(\mathcal{T} ), b(\tau)\big)$. 
Using the conditional independence structure of the model, this complete-data likelihood can be factorized as follows, 
\begin{equation}\label{Likelihood}
    \Pro_{\theta}\big({Z}, \tau, e(\mathcal{T} ), b(\tau)\big)
    \;=\; \Pro_{\theta}\big(e(\mathcal{T} )\mid {Z}, \tau, b(\tau)\big)\;
          \Pro_{\theta}\big({Z}, \tau, b(\tau) \mid Z^{t_0} \big)\;
          \Pro_{\theta}\big( Z^{t_0}\big),
\end{equation}
where $Z^{t_0} = \left\{ Z_i \right\}_{i\in V_0}$ denotes the community memberships of the individuals present in the system at time $t_0$.

\begin{proposition}\label{complete-log-lik}
The complete-data log-likelihood of the BD-SBM  is given by
\begin{align}
    \log \Pro_{\theta}\big({Z}, \tau, e(\mathcal{T} ), b(\tau)\big)
    = & \sum_{i<j} \sum_{k_1, k_2} Z_{i k_1} Z_{j k_2} 
          \!\!\! \sum_{t_\ell\in \Upsilon_{ij}} \!\! \log  \phi \big(e_{ij}(t_\ell ), k_1, k_2\big)
           + |T_B|\log  \lambda + |T_D| \log \mu \nonumber \\[0.1cm]
    -  &  (\lambda + \mu) I_N + \!\! \sum_{\tau_\ell \in T_B} \! \!  \sum_k \! \! \! \! \sum_{n=1}^{N(\tau_{\ell -1})} \!\!  \! \! L^{\ell }_{k,n+1} L^{\ell -1}_{k,n} \log n + \sum_{i \in V_{0}} \sum_k Z_{ik} \log  \beta_k, 
    \label{log_likelihood}
\end{align}
where, for any pair of distinct individuals $(i,j)$ and any pair of communities $(k_1,k_2)$,
\begin{align*}
\phi \big(e_{ij}(t_\ell), k_1 , k_2 \big) 
&=  (\pi_{k_1k_2})^{e_{ij}( t_\ell) }\big(1 - \pi_{k_1k_2}\big)^{1 - e_{ij} (t_\ell )}, \\
\Upsilon_{ij} 
&= \Big\{ t_\ell \in \mathcal{T} : t_\ell\in [\max(\tau_i^b, \tau_j^b),\, \min(\tau_i^d , \tau_j^d) ]\Big\}, \\
I_N 
&= \sum_{\ell =1}^{M} N(\tau_{\ell -1})\big(\tau_\ell-\tau_{\ell -1}\big). 
\end{align*}
\end{proposition}
Note that by construction, $\Upsilon_{ij}$ is the set of snapshot times at which both $i$ and $j$ are alive in the system. 
\begin{proof}
Under the BD-SBM setting, the three terms in the right-hand side of Equation \eqref{Likelihood} are defined as follows, 
\begin{align*}
 \Pro_{\theta}\big(e(\mathcal{T} )\mid {Z}, \tau, b(\tau)\big)
 &= \prod_{i<j} \prod_{k_1, k_2} \prod_{t_\ell \in \Upsilon_{ij}} 
     \Big[ \phi \big( e_{ij}(t_\ell ), k_1 , k_2 \big) \Big]^{Z_{ik_1} Z_{jk_2}}, \nonumber \\
 \Pro_{\theta}\big({Z}, \tau, b(\tau) \mid Z^{t_0} \big)
 &= \prod_{\ell =1}^{M} 
    \Bigg[
      e^{  -N(\tau_{\ell -1}) (\lambda + \mu)\Delta_{\ell} }
      \prod_{k=1}^{K} 
      \Big( \lambda  \sum_{i \in V(\tau_{\ell -1})}  Z_{ik} \Big)^{Z_{j_\ell k}\mathds{1} \left\{ b_\ell = 1  \right\} }
      \mu^{Z_{j_\ell k}\mathds{1} \left\{ b_\ell = -1 \right\} } 
    \Bigg]  \nonumber \\
 &= \lambda^{|T_B|}\, \mu^{|T_D|} 
    e^{- (\lambda + \mu ) I_N }
    \prod_{\tau_\ell \in T_B} \prod_{k=1}^{K} 
    \Big( \sum_{i \in V(\tau_{\ell -1})} Z_{ik} \Big)^{Z_{j_\ell k}} \nonumber \\
 &= \lambda^{|T_B|}\, \mu^{|T_D|} 
    e^{- (\lambda + \mu ) I_N }
    \prod_{\tau_\ell \in T_B}\prod_{k=1}^{K} \prod_{n=1}^{N(\tau_{\ell -1})} 
    n^{L_{k,n+1}^{\ell } L_{k,n}^{\ell -1}},  \\
 \Pro_{\theta} ( Z^{t_0}) 
 &= \prod_{i \in V_{0}} \prod_{k}  \beta_k^{Z_{ik}},\nonumber 
\end{align*}
with $\Delta_{\ell} = \tau_{\ell} - \tau_{\ell -1}$. \\
Summing the logarithm of these three terms gives the result.

\end{proof}

\subsection{Parameter estimation of the birth-death process }\label{sec:b-d_inf}

From Equation \eqref{log_likelihood}, it is apparent that the contributions of $\lambda$ and $\mu$ to the complete-data likelihood depend only on the observed birth and death times and on the total population size at each event time, but not on the latent community labels ${Z}$. Straightforward maximization of the log-likelihood with respect to $(\lambda,\mu)$ yields the following maximum likelihood estimators:
\[
\widehat{\lambda} = \frac{|T_B|}{I_N} 
\quad \text{and} \quad 
\widehat{\mu} = \frac{|T_D|}{I_N},
\]
with 
\[
I_N = \sum_{\ell =1}^{M} N(\tau_{\ell -1})\Delta_{\ell}, 
\qquad \Delta_{\ell} = \tau_\ell-\tau_{\ell-1}.
\]
These coincide with the classical ML estimators for a homogeneous birth-death process with constant birth and death rates that can be found, for instance, in \cite{reynolds1973estimating, keiding1975maximum}.

In contrast, the maximum likelihood estimation of the community proportions $\beta_k$ and block connection probabilities $\pi_{k_1k_2}$, based on the marginal distribution of the observed data alone, is neither computationally tractable nor available in closed form. We therefore resort to a VEM-type algorithm to estimate $(\pi,\beta)$ together with the latent variables (see Section~\ref{VEM}).

In what follows, we denote by 
\[
\tilde{\theta} = (\widehat{\lambda},\widehat{\mu},\pi, \beta),
\]
the parameter vector of the model composed by the ML estimators of $(\lambda,\mu)$, complemented with the remaining parameters $(\pi,\beta)$ that are still to be estimated.

\subsection{Parameter estimation for the BD-SBM and prediction of latent variables}
\label{VEM}

Classical EM algorithms rely on the posterior distribution of the latent variables given the observations, that is,
\[
\Pro_{\theta}\big({Z}  \mid  \tau, e(\mathcal{T} ), b(\tau)\big)
= \Pro_{\theta}\big( (Z_{i})_{i}  \mid  \tau, e(\mathcal{T} ), b(\tau)\big).
\]
In our model, as in standard SBMs, this posterior distribution does not factorize over individuals: for any pair of distinct nodes $i$ and $j$, and any pair of communities $(k_1,k_2)$, the indicators $Z_{ik_1}$ and $Z_{jk_2}$ are no longer independent when conditioning on the data. As a result, the exact posterior is not tractable and the E-step of EM cannot be performed in closed form.

Instead, we derive a variational expectation-maximization algorithm (VEM), a particular instance of variational inference (see \cite{bishop2006pattern}). The idea is to replace the intractable posterior by a tractable family of distributions $\mathcal{Q}$ on ${Z}$, and to optimize within this family a lower bound on the observed-data log-likelihood. 

In the sequel, $q \in \mathcal{Q}$ denotes either a distribution or its associated likelihood function. The notation $\mathbb{E}_q[\cdot]$ refers to expectation with respect to $q$.
Next, for any  $q \in \mathcal{Q}$, we decompose the observed-data log-likelihood as follows, 
\[
\log \Pro_{\tilde{\theta}}\big(  \tau, e(\mathcal{T} ), b(\tau) \big)
= \mathcal{L}(q) + \mathrm{KL}\!\left( q \,\middle\|\, \Pro_{\tilde{\theta}}(\cdot \mid \tau, e(\mathcal{T} ), b(\tau)) \right),
\]
where
\[
\mathcal{L}(q) 
= \mathbb{E}_q \left[ \log \Pro_{\tilde{\theta}}\big({Z}, \tau, e(\mathcal{T} ), b(\tau)\big) \right]
  - \mathbb{E}_q \big[ \log q ({Z}) \big],
\]
is the evidence lower bound (ELBO), and $\mathrm{KL}(\cdot\|\cdot)$ denotes the Kullback-Leibler divergence. Both can be written as follows,

\begin{align}
KL \left( q \parallel \Pro_{\tilde{\theta}} \right) & = - \mathbb{E}_q \left( \log\left( \frac{\Pro_{\tilde{\theta}}({Z} \mid \tau, e(\mathcal{T} ), b(\tau))}{q({Z})} \right) \right), \nonumber \\
\mathcal{L}(q) &=\mathbb{E}_q \left(\log \left( \frac{\Pro_{\tilde{\theta}}({Z}, \tau, e(\mathcal{T} ), b(\tau))}{q({Z})} \right) \right) \nonumber \\
& = \mathbb{E}_q \left( \log \Pro_{\tilde{\theta}} (\tau, e(\mathcal{T} ), b(\tau), {Z} ) \right) -  \mathbb{E}_q \left( \log q ({Z}) \right). \label{ELBO}
\end{align}

Since the KL divergence is non-negative, we have
\[
\log \Pro_{\tilde{\theta}}\big(  \tau, e(\mathcal{T} ), b(\tau) \big)
\;\geq\; \mathcal{L}(q),
\]
with the equality holding if and only if $q({Z})$ coincides with the true posterior distribution $\Pro_{\tilde{\theta}}({Z} \mid  \tau, e(\mathcal{T} ), b(\tau))$. The idea of the VEM approach is to maximize the ELBO $\mathcal{L}(q)$ 
with respect to  $q \in \mathcal{Q}$, which is equivalent to minimize the Kullback-Leibler (KL) divergence between $q$ and the true posterior distribution with respect to $q \in \mathcal{Q}$. Without  restriction on the distribution family  $\mathcal{Q}$, the ELBO function is maximal  for  $q({Z})=\Pro_{\tilde{\theta}}({Z} \mid  \tau, e(\mathcal{T} ), b(\tau))$, as the KL divergence becomes zero. However in the considered model, since  the posterior distribution is  not tractable, the goal is to approximate it on a  given family of variational distributions $\mathcal{Q}$ which are  tractable.  

\subsubsection{Variational distribution family $\mathcal{Q}$}\label{subsub:Var_dist}

We now specify the family of variational distributions $\mathcal{Q}$ used in our VEM algorithm. A fully factorized approximation, assuming independence among all individuals, would disregard the strong constraints that the birth-death process imposes on community sizes, and would also increase the computational cost by a factor of $N$. Instead, we adopt a structured variational approximation that explicitly accounts for the joint evolution of community sizes and, through them, induces dependencies among community memberships, while still remaining computationally tractable.

Let us define  the family of  distributions $\mathcal{Q}$ that we consider. We decompose the variational distribution $q$ on the latent labels into a product of two terms, one is related exclusively to the labels of the individuals present at $t_0$, the other one represents the conditional distribution of the labels  of the remaining individuals (those that arrive after $t_0$, that we call in the sequel the \textit{newborns}) given the size at $t_0$ of the $K$ communities,  
\begin{align}
q({Z}) &  =  q\left(\left\{ Z_i \right\}_{ i \notin  V_0} \mid Z^{t_0}\right) \; q(Z^{t_0})=  q\left(\left\{ Z_i \right\}_{ i: \tau_i^b \in T_B} \mid Z^{t_0}\right) \; q(Z^{t_0}),  \label{q_var_dist}
\end{align}
where $Z^{t_0} = \left\{ Z_i \right\}_{i\in V_0}$ are the labels of individuals present at time $t_0$ and
$\left\{ Z_i \right\}_{ i : \tau_i^b \in T_B}$ are the labels of individuals entering the system after $t_0$.

\begin{enumerate}
\item \emph{Individuals present at time $t_0$.}  
We assume independence between  the label vectors of individuals present  at the initial time, 
\begin{align}
q(Z^{t_0}) 
&= \prod_{i \in V_{0} } q(Z_i) 
 = \prod_{i \in V_{0} } \prod_{k=1}^K \delta(i,k)^{Z_{ik}},
\label{aprox_distribution}
\end{align}
where $\delta(i,k) = q(Z_{ik}=1)$ are variational parameters in $\mathcal{Q}$ and satisfy $\sum_{k} \delta(i,k) = 1$ for all $i \in V_0$.

The distribution of the number of individuals at time $t_0$ across the $K$ communities 
can be viewed as a generalization of the multinomial distribution, known as the 
\emph{Poisson Multinomial Distribution (PMD)}, see \cite{lin2023computing}. 
Unlike the classical multinomial case, here the success probabilities are not identical 
across individuals: each individual $i \in V_{0}$ is associated with a probability vector 
$\delta(i,k)$, with
\[
\sum_{k=1}^K \delta(i,k) = 1, \quad \forall i.
\]

Let $N_1, N_2, \dots, N_K$ denote the community counts at time $t_0$. By construction, the random variables satisfy the constraint
\[
\sum_{k=1}^K N_k = N_0,
\]
where $N_0$ is the population size at time $t_0$.

Our interest lies in the marginal distributions of the components $N_1, N_2, ..., N_K$.  For each community $k$ and count $n$, we define the marginal distribution by
\[
\gamma_{\text{mar}}(0,k,n) = q\big( N_k = n \big) = q\big(L_{k,n}^0 = 1\big),
\]
which corresponds to a \emph{Poisson binomial distribution} (see
\cite{tang2023poisson}) with success probabilities $\{\delta(i,k)\}_{i \in V_0}$.
More precisely,

\begin{align}\label{PMD}
q\big(L_{k,n}^0 = 1\big) & =   \sum_{A_n \subset V_0:  |A_n|=n \;} \prod_{i \in A_n}  \delta(i,k)  \prod_{j \in A_n^{c}} (1-\delta(j,k)),
\end{align}  
where $A_n^{c}$ denotes the complementary of $A_n$ in $V_0$. 

\item \emph{Individuals born after $t_0$ (newborns).} The community memberships of individuals not present at $t_0$ are modelled conditionally on $Z^{t_0}$ as
\begin{align}
q\left( \left\{ Z_i \right\}_{ i \notin  V_0} \mid Z^{t_0}\right) =  q\left(\left\{Z_i\right\}_{ i: \tau_i^b \in T_B} \mid Z^{t_0}\right),  \label{markov_0}
\end{align}
Given the community memberships of individuals present at $t_0$, the joint distribution of the community assignments of all newborns is fully characterised by the joint distribution of the community sizes at the event times, that is
\begin{align*}
\left\{Z_i\right\}_{ i: \tau_i^b \in T_B} \mid Z^{t_0} & \stackrel{\mathcal{D}}{=  } (L^{1}, \ldots, L^{M}\mid Z^{t_0}),
\end{align*}   
where $L^{\ell } = (L_{k,n}^{\ell })_{k,n}$ denotes the collection of community-size indicators at the event time $\tau_\ell$. 

Since each event time corresponds to either a single birth or a single death, the sequence of
community sizes $(L^{0},\ldots,L^{M})$ forms a Markov chain in the generative model: at any event
time $\tau_\ell$ the distribution of $L^{\ell}$ given the past depends on the history only through $L^{\ell-1}$. However, in our inferential setting we condition on the full event history
$(\tau, b(\tau))$, which records not only whether a birth or a death occurs at each $\tau_\ell$,
but also the identity $i_\ell$ of the individual involved in the event. This additional information breaks the Markov property at death times: when $b_\ell=-1$,  $L^\ell$ depends on the past not only through $L^{\ell-1}$ but also through the latent memberships $Z_i$ of individuals alive before
$\tau_\ell$.

At $\tau_\ell$, the knowledge of the past $(L^{1:\ell -1}, Z^{t_0})$ is equivalent to the knowledge of $(L^{\ell -1}, \left\{ Z_i\right\}_{i \in V(\tau_0:\tau_{\ell -1})})$, where $L^{1:\ell -1}=(L^{1},\ldots,L^{\ell -1})$ and $V(\tau_0:\tau_{\ell -1}) = \bigcup_{j=0}^{\ell -1} V(\tau_j)$.  Therefore, we have under $q$

\[ L^{\ell} \mid  L^{1:\ell -1}, Z^{t_0}   \stackrel{\mathcal{D}}{=} L^{\ell} \mid  L^{\ell -1}, \left\{ Z_i\right\}_{i \in V(\tau_0:\tau_{\ell -1})}. \]

At birth times ($\tau_\ell \in T_B$), the Markov property gives rise to the following simplification
\[
L^{\ell} \mid  L^{1:\ell -1}, Z^{t_0}
\;\stackrel{\mathcal{D}}{=}\;
L^{\ell} \mid  L^{\ell -1}.
\]
Under our variational approximation, Equation \eqref{markov_0} is therefore factorized  as follows, 

\begin{align*}
q\left( \left\{ Z_i \right\}_{ i \notin  V_0} \mid Z^{t_0}\right)  & = q \left(L^{1}, \ldots, L^{M}\mid Z^{t_0} \right) \nonumber \\
    & =  \prod_{\ell =1}^{M} q\left(L^{\ell } \mid  L^{\ell -1} , \left\{ Z_i\right\}_{i \in V(\tau_0:\tau_{\ell -1})}\right)   \nonumber   \\
      & =  \prod_{\ell =1}^{M} q\left(L^{\ell } \mid  L^{\ell -1}\right)^{\mathds{1} \left\{ b_\ell=1  \right\} } q\left(L^{\ell } \mid  L^{\ell -1} , \left\{ Z_i \right\}_{i \in V(\tau_0:\tau_{\ell -1})}\right)^{\mathds{1}\left\{b_\ell=-1  \right\}}.
\end{align*} 

At death times, once the identity of the departing individual is known, the updated community sizes
are deterministically obtained from $L^{\ell -1}$ and $\left\{ Z_i \right\}_{i \in V(\tau_0:\tau_{\ell -1})}$, so that the
conditional distribution reduces to,

\[ q\left(L^{\ell } \mid  L^{\ell -1} , \left\{ Z_i \right\}_{i \in V(\tau_0:\tau_{\ell -1})}\right)^{\mathds{1} \left\{ b_\ell=-1  \right\}} = 1. \]

We can then rewrite Equation \eqref{markov_0} as follows, 
\begin{equation}
    q\big(\left\{ Z_i \right\}_{ i \notin  V_0} \mid Z^{t_0}\big)
    = \prod_{\tau_\ell\in T_B} \prod_{k=1}^K \prod_{n = 1}^{N(\tau_{\ell-1})}
      \gamma(\ell,k,n,n)^{L_{k,n}^{\ell-1} L_{k,n}^{\ell}}\,
      \gamma(\ell,k,n,n+1)^{L_{k,n}^{\ell-1} L_{k,n+1}^{\ell}},
    \label{q_distrib} 
\end{equation}
where $\gamma(\ell,k,n,n')$ denotes the transition probability to move from $n$ to $n'$ individuals in
community $k$ at time $\tau_\ell$, 
that is
\[
  \gamma(\ell,k,n,n')
  \;=\;
  q\big( L_{k,n'}^{\ell} = 1 \mid L_{k,n}^{\ell-1} = 1 \big),
\]
which are  constrained by
\begin{align*}
\gamma(\ell,k,n,n+b_\ell) + \gamma(\ell,k,n,n) & = 1 \quad \text{for all } \ell \geq 1,\\
\text{and} \quad \gamma(\ell,k,n,n')  & = 0 \quad \text{for all } n' \notin \{n,\,n+b_\ell\}, \; \quad \text{for all } \ell \geq 1.
\end{align*}

\medskip 

To summarize, the variational distribution $q$ in \eqref{q_var_dist} is fully specified by Equations~\eqref{aprox_distribution} and~\eqref{q_distrib}, and is parametrized by $\delta(i,k)$ for all $i \in V_0$, $k=1,\dots,K$, and by the transition probabilities $\gamma(\ell,k,n,n')$ for all $k$, $n$, $n'$ and event times $\tau_\ell \in T_B$. Intuitively, the $\delta(i,k)$'s control the community memberships at the initial time $t_0$, while the $\gamma(\ell,k,n,n')$'s govern the evolution of community sizes at subsequent birth events.

Once the transition probabilities $\gamma(\ell,k,n,n')$ have been introduced, we need to relate them to the membership probabilities $\delta(i,k)$ for individuals born after $t_0$. To this end, we work with the marginal distributions of the community sizes at all event times $\tau_\ell \in \tau$. In what follows, we first define these marginal probabilities $\gamma_{\text{mar}}(\ell,k,n)$ and then express the probabilities $\delta(i_\ell,k)$ for individuals $i_l$, who  arrive at $\tau_l$  in the system, in terms of $\gamma_{\text{mar}}$ and $\gamma$.

\begin{enumerate} 
\item \emph{Marginal distributions of community sizes.}
For each $\ell$, $k$ and $n$, we define the marginal distribution of $L_{k,n}^{\ell}$ by
\[
\gamma_{\text{mar}}(\ell,k,n) = q\big(L_{k,n}^{\ell} = 1\big),
\]
which satisfies the recursion
\begin{align}
\gamma_{\text{mar}}(\ell,k,n)
&= \sum_{n' \in \{n,\,n-b_\ell\}} \gamma(\ell,k,n',n)\,
   \gamma_{\text{mar}}(\ell-1,k,n'),
\qquad \forall \ell \geq 1,\; k,\; n, 
\label{marginal}
\end{align}
and which have been already  introduced  for $t_0$, by   $\gamma_{\text{mar}}(0,k,n)$,  defined in \eqref{PMD}.
Note that the ELBO does not explicitly depend on the transition probabilities
$\gamma(\ell,k,n,n')$ at event times $\tau_{\ell} \in T_D$. Nevertheless, the marginals $\gamma_{\text{mar}}(\ell,k,n)$ are
required at every event time (both births and deaths) in order to apply the recursion
\eqref{marginal} and propagate information forward in time, so the full set of transitions
is needed to compute all $\gamma_{\text{mar}}(\ell,k,n)$.

\item \emph{Community membership probabilities induced by $\gamma$ and $\gamma_{\text{mar}}$.}
Once probabilities  $\gamma$ and $\gamma_{\text{mar}}$ 
are specified, we can express the variational community membership probabilities at birth times in terms of these transitions.
The key observation is that the event “the newborn $i_\ell$ at time $\tau_\ell$ belongs to community $k$’’ exactly coincides with an increase by one unit in the size of community $k$ at $\tau_\ell$. Therefore, the quantity $\delta(i_\ell,k)$ can be written as
\begin{align}
\delta(i_\ell, k)
&= \sum_{n=1}^{N(\tau_{\ell-1})} \gamma(\ell,k,n,n+1)\,
   \gamma_{\text{mar}}(\ell-1,k,n) \nonumber \\
&= \sum_{n=1}^{N(\tau_{\ell-1})} \big[1 - \gamma(\ell,k,n,n)\big]\,
   \gamma_{\text{mar}}(\ell-1,k,n), \nonumber
\end{align}
where $i_\ell$ denotes the newborn at time $\tau_\ell$ and
$\gamma_{\text{mar}}(\ell-1,k,n)$ is the marginal probability that community $k$ has size
$n$ just before $\tau_\ell$.

\end{enumerate}
\end{enumerate}

\subsubsection{Variational expectation maximization inference}\label{sec:vem_inf}
Proposition \ref{complete-log-lik} and the family of variational distributions lead to the following expression for the ELBO function given by Equation \eqref{ELBO}:

\begin{align}
   \mathcal{L}(q) & =  \mathbb{E}_q \left( \log( \Pro_{\tilde{\theta}} \left( {Z}, \tau, e(\mathcal{T} ), b(\tau)\right)  \right) -  \mathbb{E}_q \left( \log q ({Z}) \right)  \nonumber \\
     &  = \sum_{i<j} \sum_{k_1, k_2} \delta(i,k_1) \delta(i,k_2) \sum_{\ell \in \Upsilon_{ij}} \log \phi (e_{ij}(t_\ell ), k_1, k_2) +  |T_B|\log \widehat{\lambda} + |T_D| \log \widehat{\mu}- (\widehat{\lambda} + \widehat{\mu}) I_N  \nonumber \\
     & \quad  + \sum_{\tau_\ell \in T_B} \sum_k \sum_{n=1}^{N(\tau_{\ell-1})} \gamma(\ell,k,n,n+1)\gamma_{mar}(\ell -1, k, n)  \log  n  + \sum_{i \in V_{0}} \sum_k \delta(i,k)  \log \beta_k \nonumber \\
    & \quad - \sum_{i\in V_0} \sum_k \delta(i,k) \log \delta(i,k) - \sum_{\tau_\ell \in T_B} \sum_k \sum_{n=1}^{N(\tau_{\ell-1})} \Big[ \gamma(\ell,k, n, n) \gamma_{\text{mar}}(\ell -1,k, n)\log \gamma(\ell,k, n, n) \nonumber\\
     & \quad \quad  \quad + \gamma(\ell,k, n, n+1) \gamma_{\text{mar}}(\ell -1,k, n)\log \gamma(\ell,k, n, n+1)\Big].\label{expectation}
\end{align}

maximizing $\mathcal{L}(q)$ over $q\in\mathcal{Q}$ is equivalent to optimizing over the variational parameters $\delta(i,k)$ for $i\in V_0$ and $\gamma(\ell,k,n,n')$ under the constraints

 \begin{align}\label{constraints}
\left\{ 
\begin{array}{ll}
 (a) &  \sum_{ k } \delta(i,k) = 1, \quad   \forall i, \\
 & \\
(b)   &  \sum_{k}\sum_{n = 0}^{N(\tau_{\ell -1})} \; \gamma(\ell,k, n, n)  \; \gamma_{mar}(\ell -1,k,n)  =K-1, \quad \forall \ell \geq 1,\\ 
&\\ 
(c) & \gamma(\ell,k,n,n+b_\ell) + \gamma(\ell,k,n,n)  =  1,  \quad  \quad \forall \ell \geq 1. \\
 \end{array} \right.
 \end{align}  
 
Note that  relations (b) and (c) imply, 
\[\sum_{k}\sum_{n = 0}^{N(\tau_{\ell -1})} \; \gamma(\ell,k, n, n+b_\ell)   \;\gamma_{mar}(\ell -1,k,n) = 1, \quad \forall \ell \geq 1, \]
which expresses the fact that exactly one community size changes at each event time.

\medskip

Adapting  Variational Expectation Maximization algorithms to our setting  leads to the  pseudo-code in Algorithm 1. 

\begin{algorithm}[h]
\caption{Variational Expectation Maximization}
\label{pseudocode}

\begin{algorithmic}[1]\label{algo:VEM_BD-SBM}
\State \textbf{Input:} Observed data $(\tau, e(\mathcal{T}), b(\tau))$; initial parameters $(\pi^{(0)},\beta^{(0)})$; initial variational distribution $q^{(0)}$; pre-estimated birth-death parameters $(\widehat{\lambda},\widehat{\mu})$.
\State \textbf{Output:} Estimated parameters $(\pi,\beta)$ and optimal variational distribution $q \in \mathcal{Q}$.
\State Set iteration counter $m = 0$
\Repeat
    \State \textbf{VE-Step:} \Comment{Update the variational distribution $q^{(m+1)}$}
    	\State \quad Maximize the ELBO with respect to $q$
    \[
    q^{(m+1)} = \arg\max_{q \in \mathcal{Q}} \left\{ \mathbb{E}_{q}[\log \Pro_{\tilde{\theta}^{(m)}}( {Z},\tau, e(\mathcal{T} ), b(\tau))] - \mathbb{E}_{q}[\log q({Z})] \right\}
    \]
    \State \textbf{VM-Step:} \Comment{Update the model parameters $\tilde{\theta}^{(m+1)}=(\widehat{\lambda},\widehat{\mu},\pi^{(m+1)},\beta^{(m+1)})$}
    	\State \quad Maximize the ELBO with respect to $(\pi,\beta)$ a part of    $\tilde{\theta}=(\widehat{\lambda},\widehat{\mu},\pi,\beta)$:
    \[
    (\pi^{(m+1)},\beta^{(m+1)}) = \arg\max_{(\pi,\beta)} \mathbb{E}_{q^{(m+1)}}[ \log \Pro_{\tilde{\theta}} ({Z}, \tau, e(\mathcal{T} ), b(\tau)) ]
    \]
    \State Increment iteration counter $m = m + 1$
\Until{a convergence criterion is met.}
\end{algorithmic}
\end{algorithm}

\bigskip

In the VE-step, we must determine the $\delta(i,k)$ for all $i\in V_0$ and the $\gamma(\ell,k,n,n)$ for all birth times $\tau_\ell\in T_B$ that maximize $\mathcal{L}(q)$ under the constraints \eqref{constraints}. Recall that under $q$, $\delta(i,k)$ is the variational community membership distribution for individual $i$, whereas $\gamma(\ell,k,n,n)$ encodes the transition probabilities for community-size changes.

\begin{proposition} \label{prop:gamma}
For $\ell \in T_B$, we refer to $i_\ell$ as the newborn at time $\tau_\ell$, and we denote by $\widehat{\gamma}(\ell,k,n,n)$ the quantity that maximizes $\mathcal{L}(q)$ with respect to $\gamma(\ell,k,n,n)$ when $\tilde{\theta} = (\widehat{\lambda},\widehat{\mu},\pi,\beta)$ is fixed. Then there exists a unique $\rho_\ell > 0$ such that
\[
\widehat{\gamma}(\ell,k,n,n)
= \frac{\rho_\ell}{
    \rho_\ell
    + n \displaystyle\prod_{j \neq i_\ell}
                       \prod_{k'}
                       \prod_{\ell' \in \Upsilon_{i_\ell j}}
                       \Big[ \phi\big( e_{ij}(t_{\ell ' }), k, k' \big) \Big]^{\widehat{\delta}(j,k')}
  },
\qquad \text{for } n \in [0, N_{\ell-1}[,
\]
where the parameter $\rho_\ell$ is determined by the constraint (b) in~\eqref{constraints}.
\end{proposition}

\begin{proof}
The result follows by maximizing $\mathcal{L}(q)$ with respect to $\gamma(\ell,k,n,n)$, using Lagrange multipliers for the constraints in \eqref{constraints}. See Appendix~\ref{proof:gamma} for details.
\end{proof}

\begin{remark}
In order to compute $\widehat{\gamma}(\ell,k,n,n)$ in practice, one must first determine
$\rho_\ell$. Since the defining equation for $\rho_\ell$ has no closed-form solution in
general, $\rho_\ell$ is obtained numerically.
\end{remark}

Using Proposition~\ref{prop:gamma} together with (c) in \eqref{constraints} immediately yields, for $\ell  \in T_B$,
\[
\widehat{\gamma}(\ell,k, n, n+1) = 1  - \widehat{\gamma}(\ell,k, n, n).
\]

Optimizing $\mathcal{L}(q)$ with respect to the initialization parameters $\delta(i,k)$ for $i\in V_{0}$ is more delicate. Following the approximation strategy of \cite{matias2017statistical, agarwal2025clustering}, we neglect the dependence of certain terms in $\mathcal{L}(q)$ on $\delta(i,k)$ and update $\delta(i,k)$ by solving the fixed-point equation
\begin{equation}
\widehat{\delta}(i,k)\;\propto\; \beta_k
\prod_{j\neq i}\;\prod_{k'}\;\prod_{\ell' \in \Upsilon_{ij}}
\Bigl[\phi\!\bigl(e_{ij}(t_{\ell'}),\,k,\,k'\bigr)\Bigr]^{\widehat{\delta}(j,k')},
\label{eq:delta-fixed-point}
\end{equation}
where $\phi\!\bigl(e_{ij}(t_{\ell'}),k,k'\bigr)$ and $\Upsilon_{ij}$ are defined in Proposition~\ref{complete-log-lik}.
In simulations (Section~\ref{sec:simu}), this approximation performs well.
For completeness, we provide in the Appendix~\ref{proof:gamma} the derivative used in this update.

For $i\notin V_{0}$, the computation of $\delta(i,k)$ does not entail $\mathcal{L}(q)$ optimization. Indeed, the event that the newborn $i_\ell$ at time $\tau_\ell$ belongs to community $k$ is equivalent to the size of community $k$ increasing by one at $\tau_\ell$, we can express the variational community membership probability $\widehat{\delta}(i_\ell,k)$ in terms of the transition probabilities $\widehat{\gamma}(\ell,k,n,n+1)$ as
\begin{align}
\widehat{\delta}(i_\ell, k) 
&= \sum_{n=1}^{N(\tau_{\ell -1})}
     \widehat{\gamma}(\ell , k, n, n + 1)\, \widehat{\gamma}_{\text{mar}}(\ell -1, k, n),
\label{restriction1}
\end{align}
where\;  
\begin{align} \label{gamm_mar_rec}
\widehat{\gamma}_{\text{mar}}(\ell -1, k, n) = \left\{ 
\begin{array}{lcl}
 \displaystyle{\sum_{n' \in \{n,\,n-b_\ell\}}} \widehat{\gamma}(\ell -1,k,n',n)\,\widehat{\gamma}_{\text{mar}}(\ell -2,k,n'),  & \text{for} & \ell-1 >0,\\
  \displaystyle{\sum_{A_n \subset V_0:  |A_n|=n}} \prod_{i \in A_n}  \widehat{\delta}(i,k)  \prod_{j \in A_n^{c}} (1-\widehat{\delta}(j,k)),& \text{for} & \ell-1=0,
 \end{array} \right. \end{align}
using the PMD expression \eqref{PMD}. The term $n=0$ is excluded from the sum in \eqref{restriction1}, since we cannot have any birth event on a community once this community is empty.

We optimize the ELBO explicitly only with respect to the transition probabilities at birth times, since these are the only ones appearing directly in \eqref{expectation}. Transition probabilities at death times do not appear in the ELBO in closed form since they are already determined \textit{implicitly} by the probabilities of previous events, because we observe the individual involved in each death event. Computing them exactly would require summing over all possible past configurations of community memberships. In other words, it requires to do the \textit{deconvolution} of the $\gamma$ probabilities backwards in time up until the time of birth of the departing individual (so, as many times as the number of events that take place during the lifetime of this individual), which is intractable,   

\begin{samepage}
\begin{align*} 
 \gamma(\ell , k, n, n -1 ) & = q( L_{k, n-1}^{\ell }=1 \mid L_{k, n-1}^{\ell -1} =1) \\
 &  = \sum_{ Z_i \in \{0,1\}^K, \; \forall \; i \in V(\tau_0:\tau_{\ell -1}) }\; q( L_{k,n-1}^{\ell }=1, \left\{ Z_i \right\}_{i \in V(\tau_0:\tau_{\ell -1})} \mid L_{k, n-1}^{\ell -1} =1). 
\end{align*}
\end{samepage}

These probabilities are still needed to propagate the marginals via \eqref{gamm_mar_rec}. Therefore, instead of computing them directly, for $b_\ell=-1$ and individual $i_\ell$ dying at time $\tau_\ell$, we approximate $\gamma(\ell,k,n+1,n)$ by enforcing the additional constraint

\begin{align}
\delta(i_\ell, k) = \sum_{n=0}^{N(\tau_{\ell -1})-1} \gamma(\ell , k, n+1, n)\, \gamma_{\text{mar}}(\ell -1, k, n+1),
\label{extra}
\end{align}
which simply states that the probability that the size of community $k$ decreases by one at $\tau_\ell$ must coincide with the probability that $i_\ell$ belongs to community $k$.

Observing that $\gamma(\ell , k, n+1, n)$ should increase with $n$, we adopt the following parametric form
\begin{align}
    \gamma(\ell , k, n+1, n) = \frac{\rho_{\ell k}\, n}{1 + \rho_{\ell k}\, n}, \label{gamma_deaths}
\end{align}
where $\rho_{\ell k}\geq 0$ is a parameter specific to the $\ell$-th BD (which is a death) event and community $k$. We show in  Appendix~\ref{proof:gamma_deaths} that there is a unique solution of this form satisfying \eqref{extra}.

Let $\widehat{\rho}_{\ell k}$ being the positive numerical solution of the next equation in   ${\rho}_{\ell k}$

\[ \sum_{n=0}^{N(\tau_{\ell -1})-1}  \dfrac{ {\rho}_{\ell k} n}{1 + {\rho}_{\ell k} n}\; \widehat{\gamma}_{\text{mar}}(\ell -1, k, n+1) = \widehat{\delta}(i_\ell, k).\]
which is the constraint \eqref{extra} with $\gamma(\ell,k,n+1,n)$ replaced by its parametric expression.

\bigskip
In the VM-step of the algorithm, we update the parameters $\pi_{k_1k_2}$ and $\beta_k$ by maximizing the ELBO function $\mathcal{L}(\cdot)$ written in equation \eqref{expectation} for fixed $q$.

\begin{proposition}\label{prop:pi_beta}
Let $q \in \mathcal{Q}$ be fixed. Then the maximizers of $\mathcal{L}(q)$ 
with respect to $\pi_{k_1 k_2}$ and $\beta_k$ are given by
\[
\widehat{\pi}_{k_1 k_2} = \frac{\sum_{i \neq j} \delta(i,k_1)\, \delta(j,k_2)\; \sum_{\ell \in \Upsilon_{ij}} e_{ij}(t_\ell )}
     {\sum_{i \neq j} \delta(i,k_1)\, \delta(j,k_2)\, |\Upsilon_{ij}|},
\]
where $\Upsilon_{ij}$ is defined in  Proposition \ref{complete-log-lik}, and
\[
\widehat{\beta}_{k} 
= \frac{1}{|V_0|}\sum_{i\in V_{0}}  \delta(i,k).
\]
\end{proposition}

\begin{proof}
For a fixed $q \in \mathcal{Q}$, the result is straightforward by solving the maximization of the  ELBO function $\mathcal{L}(q)$ with respect to $\pi_{k_1 k_2}$. Similarly, maximizing $\mathcal{L}(q)$ with respect to $\beta_k$, under the constraint $\sum_k \beta_k = 1$, gives the stated formula for $\widehat{\beta}_k$. Detailed calculations are given in Appendix \ref{proof:pi_beta}.
\end{proof}

VEM was first introduced in the SBM context by \cite{daudin2008mixture}, who derived the variational updates for the static binary SBM under a mean–field variational family. Our variational expressions can be viewed as natural temporal extensions of this classical mean–field framework for stochastic block models. See, for instance, \cite{matias2017statistical} for a dynamic generalisation of the static SBM using mean–field variational inference in the case of a fixed population size, and related developments in other directions.

 Our results recover the same functional form for the updates of $\left( \delta(i,k)\right)_{i,k}$ for $i \in V_0$ and of $(\pi_{k_1k_2},\beta_k)$, but embed them in a richer temporal model where the population size is itself random and evolves according to a birth–death mechanism. In particular, the additional variational parameters $(\gamma,\gamma_{\text{mar}})$ explicitly encode the evolution of community sizes and enforce consistency between individual memberships and the latent population process.

From this perspective, our BD–SBM can be seen as a dynamic generalisation of the static variational SBM. When there are no births or deaths (so that the population size remains constant), our updates reduce to the standard VEM updates of a dynamic SBM with a fixed vertex set.

\subsection{Model selection with the Integrated Completed Likelihood (ICL)}\label{sec:ICL}

A key aspect of fitting a Stochastic Block Model is the choice of $K$, the number of latent communities, which is generally unknown. Using too few classes may oversimplify the network structure, whereas using too many may lead to overfitting. To select an appropriate value of $K$, we rely on the Integrated Completed Likelihood (ICL) criterion \cite{biernacki2000assessing,daudin2008mixture}.

For any number of communities $K \geq 1$, we make explicit the dependence of the model parameters on $K$ by writing $\theta_K = \{\lambda,\mu, \beta_K,\pi_{K}\}$. For fixed $K$, the corresponding estimator is denoted by $\widehat{\theta}_K=(\widehat{\lambda}, \widehat{\mu},\widehat{\beta}_K, \widehat{\pi}_K)$, as described in Sections~\ref{sec:b-d_inf} and~\ref{sec:vem_inf}.

Given $K$, we denote by $\widehat{{Z}}$ and $\widehat{L}$ the maximum a posteriori (MAP) estimates of the community memberships and community sizes, respectively, under the VEM. The complexity of our model, as a function of $K$, is driven by:
\begin{enumerate}
\item the $K-1$ degrees of freedom associated with the proportion parameters
      $\beta_K$ for the $|V_0|$ individuals present in the system at time $t_0$;
\item the $\tfrac{K(K+1)}{2}$ block connection probabilities
      $\pi_{k_1k_2}$ governing the $\sum_{i<j} |\Upsilon_{ij}|$ Bernoulli trials;
\item the birth–death parameters $(\lambda,\mu)$, which do not depend on $K$ and are
      therefore not included in the penalisation term.
\end{enumerate}

The ICL criterion is then defined as follows, 
\begin{equation*}
    \mathrm{ICL}(K)
    = \log \Pro_{\widehat{\theta}_K} \left( e(\mathcal{T} ), \widehat{{Z}}, \tau, b(\tau) \right)
    - \tfrac{1}{2} (K-1) \log N_0
    - \tfrac{1}{2} \tfrac{K(K+1)}{2}
 \log \Bigg( \sum_{i<j} |\Upsilon_{ij}| \Bigg).
\end{equation*}

Maximizing $\mathrm{ICL}(K)$ over candidate values of $K$ provides a commonly use for selecting the number of communities. In practice, we fit the model for a range of values of $K$, compute the ICL for each, and select the value that maximizes the criterion. This approach balances model fit and complexity while explicitly accounting for the uncertainty in the community assignments.

\subsection{Initialization of the Variational EM Algorithm}\label{initialization}

The initialization of the latent variables is a crucial step for the variational expectation-maximization (VEM) algorithm, since poor starting values may lead to convergence towards suboptimal local maxima. We describe here the procedure used to initialize the community membership probabilities $\delta(i,k)$ for all $i \in V$ and $k=1,\ldots,K$. Throughout this subsection, $\delta^{(0)}(i,k)$ denotes the initial value of $\delta(i,k)$.

To obtain a first partition of the nodes into $K$ communities, we build a similarity score between each pair of nodes $(i,j)$, defined by 
\[
s_{ij} = \frac{2\sum_{\ell \in \Upsilon_{ij}} e_{ij}(t_\ell ) - |\Upsilon_{ij}|}{|\Upsilon_{ij}|},
\]
which measures the average tendency of nodes $i$ and $j$ to be connected during the time interval when they are simultaneously present in the network. Let $S = (s_{ij})$ denote the resulting similarity matrix. We apply a $K$-means clustering algorithm to $S$ and obtain an initial hard partition of the nodes into $K$ groups. The associated binary indicators are denoted by $\widehat{Z}^{(0)}_{ik}$, where $\widehat{Z}^{(0)}_{ik} = 1$ if node $i$ is assigned to cluster $k$, and $0$ otherwise.

We then construct the initial variational membership probabilities $\delta^{(0)}(i,k)$ so as to be consistent with this hard clustering while allowing for uncertainty. For a fixed weight $\omega \in (0,1)$, we set
\[
\delta^{(0)}(i,k)
= \omega\, \widehat{Z}^{(0)}_{ik} + \frac{1-\omega}{K}, 
\qquad i \in V,\; k=1,\ldots,K.
\]
Thus each node receives a larger probability $\omega + (1-\omega)/K$ for its assigned cluster and a smaller, uniform probability $(1-\omega)/K$ for the remaining communities.

We decided to initialize the model parameters $\beta$ and $\pi$ according to the following formulas.  

\begin{enumerate}
\item The initial community proportions are given by
\[
\beta^{(0)}_{k}
= \frac{1}{|V_0|}\sum_{i \in V_{0}} \delta^{(0)}(i,k),
\]
that is, the average assignment probability of individuals in $V_{0}$ to community $k$.
\item The initial block connection probabilities are set to
\[
\pi^{(0)}_{k_1 k_2}
= \frac{\sum_{i<j} \delta^{(0)} (i,k_1)\, \delta^{(0)} (j,k_2)\, \sum_{\ell  \in \Upsilon_{ij}} e_{ij}(t_\ell )}
       {\sum_{i<j} \delta^{(0)} (i,k_1)\, \delta^{(0)} (j,k_2)\, |\Upsilon_{ij}|}.
\]
\end{enumerate}

This initialization scheme provides informative starting values that reflect the empirical clustering structure of the network, while preserving enough flexibility for the VEM algorithm to refine the variational distribution in subsequent iterations. In particular, the initialization does not enforce the birth-death structure on community sizes, and the constraints induced by the birth-death dynamics are only imposed through the subsequent VEM updates.

\section{Experiments}

We evaluate the proposed model through two complementary types of experiments: two controlled simulated data sets and a real-world temporal  co-authorship network built from \texttt{arXiv}. An open-source Python implementation of the methods described in this paper is available at this \href{https://github.com/gbayolo26/dynamic-bdsbm-vem}{GitHub repository}.

With simulated data, we assess the performance of our model through clustering accuracy, the recovery of group memberships, and the accuracy of parameter estimation. The \texttt{arXiv} application demonstrates how the model behaves on a large, sparse network with heterogeneous communities and genuine temporal turnover, providing evidence of its practical relevance.

\subsection{Simulated data}\label{sec:simu}

We generate two synthetic data sets, one according to our birth-death Stochastic Block Model (BD-SBM) and the other one simulated from a pure birth process embedded within  a dynamic Stochastic Block Model (B-SBM).

For the BD-SBM, the birth-death process is simulated with birth rate $\lambda = 0.04$ and death rate $\mu = 0.02$ starting at $t_0=0$ and end up at $t_{T} = 150$. We consider $K = 4$ latent communities with initial sizes $10$, $11$, $7$, and $12$, yielding an initial population $N_0 = 40$. 

At $t_{T} = 150$, the total number of distinct individuals who have been involved in the process reaches $N = 1208$.

For the pure birth process within a dynamic SBM ("B-SBM"), the birth process is  generated a according to the same parameters as for the birth-death process described above except  that the death rate is set to $\mu = 0$ and the horizon is $t_{\text{end}} = 100$, while keeping $\lambda = 0.04$. 
In this setting, At $t_{T} = 100$, the total number of distinct individuals who have been involved in the process reaches $N = 2601$.

For both scenarios, we construct two dynamic SBMs by specifying different within- and between-community connection probabilities. 
The two regimes are encoded by the following symmetric block matrices:

\[
\pi_{\text{low}} =
\begin{pmatrix}
0.05 & 0.09  & 0.05  & 0.04 \\
0.09 & 0.10  & 0.055 & 0.06 \\
0.05 & 0.055 & 0.20  & 0.07 \\
0.04 & 0.06  & 0.07  & 0.06
\end{pmatrix}
\qquad
\pi_{\text{high}} =
\begin{pmatrix}
0.75 & 0.36 & 0.20 & 0.16 \\
0.36 & 0.91 & 0.22 & 0.24 \\
0.20 & 0.22 & 0.82 & 0.28 \\
0.16 & 0.24 & 0.28 & 0.66
\end{pmatrix}
\]
The first matrix, $\pi_{\text{low}}$, corresponds to a low-signal regime in which within- and between-community probabilities are relatively close (resulting in a sparser and less clearly separated network connection probabilities), whereas $\pi_{\text{high}}$ represents a high-signal regime with much stronger within-community connectivity.

Network snapshots are recorded at equally spaced times over the interval $[t_0=0, t_{T}]$, with one snapshot per unit of time, resulting in $150$ snapshots for the BD-SBM and $100$ snapshots for the pure birth model.

Note that in our implementation, a single VEM iteration requires 
\[
 \mathcal{O}\bigl(K K + M K N \bigr)= \mathcal{O}\bigl(M K N \bigr)\quad \text{operations, since usually} \; M N >> K,
\]
where $M$ is the number of birth-death events, $N=|V|= |\cup_{t \in [t_0,t_M]}\;V(t)|$ is the total population size over the observation window, and $K$ is the number of communities. 
This cost is dominated by the updates of the pairwise edge contributions and community responsibilities across all event times.

\subsubsection{Birth-death process within dSBM}\label{BDSBM_experiments}

From the simulated dataset, we first choose the optimal $K$ using  the ICL criterion described in section \ref{sec:ICL}.
In particular, we perform $50$ random initializations of the VEM algorithm and, for each of them, we retain the value of $K$ that maximizes the  ICL.
Figure~\ref{fig:ICL_BDSBM} reports, for each regimes, how often each value of $K$ is selected across the random runs. For both connectivity regimes, the best trade-off between goodness-of-fit and model complexity is obtained for $K = 4$, which coincides with the true number of communities used to generate the data. 
Second,  we apply the proposed VEM algorithm with $K = 4$. 

\begin{figure*}[h!]
  \centering
  \begin{subfigure}[b]{.45\textwidth}
    \centering
    \includegraphics[width=0.8\linewidth]{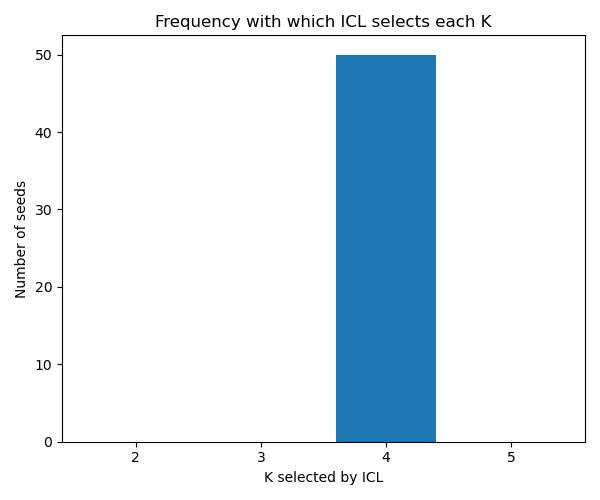}
    \caption{$\pi_{\text{low}}$.}
    \label{fig:ICL_BDSBM_pi_small}
  \end{subfigure}
  \hfill
  \begin{subfigure}[b]{.45\textwidth}
    \centering
    \includegraphics[width=0.8\linewidth]{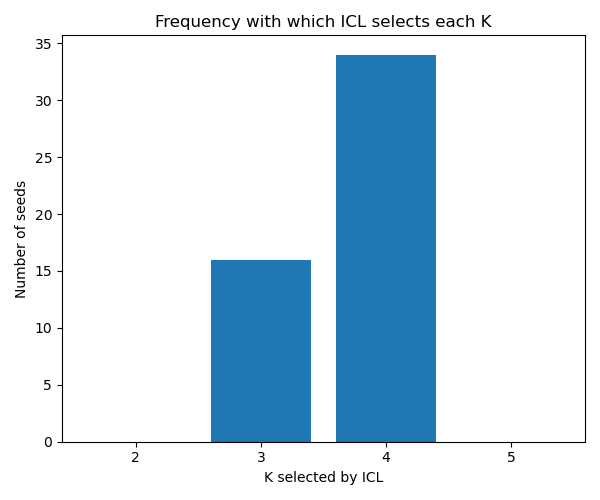} 
    \caption{$\pi_{\text{high}}$.}
    \label{fig:ICL_BDSBM_pi_big}
  \end{subfigure}
  \caption{Number of times the ICL criterion selects each value of $K$ as the optimal number of communities over all random initializations in the BD-SBM experiment with $\pi_{\text{low}}$ (left) and $\pi_{\text{high}}$ (right).}
  \label{fig:ICL_BDSBM}
\end{figure*}

As shown in Figure~\ref{fig:ELBO_BDSBM}, our algorithm converges quickly, in fewer than   $10$ iterations. As expected, the algorithm converges more rapidly under the high-signal regime $\pi_{right}$ than under the low-signal regime $\pi_{low}$.  Since our implementation relies on an approximate version of the VEM,
the classical monotonicity property of the ELBO cannot be theoretically guaranteed. Nevertheless, in practice, the ELBO trajectories are nearly monotonic, rapidly reach a plateau, and lead to highly accurate parameter estimates.

\begin{figure*}[h!]
  \centering
  \begin{subfigure}[b]{.45\textwidth}
    \centering
    \includegraphics[width=0.8\linewidth]{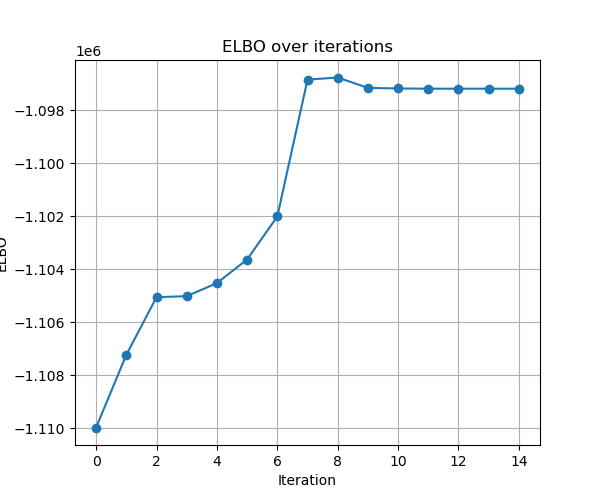}
    \caption{$\pi_{\text{low}}$.}
    \label{fig:ELBO_BDSBM_pi_small}
  \end{subfigure}
  \hfill
  \begin{subfigure}[b]{.45\textwidth}
    \centering
    \includegraphics[width=0.8\linewidth]{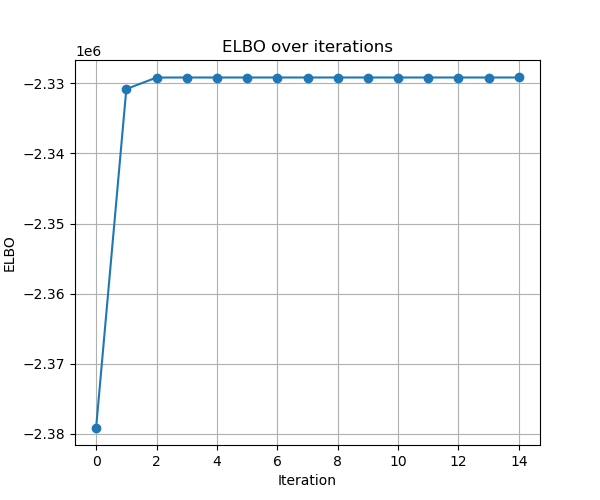} 
    \caption{$\pi_{\text{high}}$.}
    \label{fig:ELBO_BDSBM_pi_big}
  \end{subfigure}
  \caption{Evolution of the ELBO over VEM iterations for the two simulated BD-SBM networks with $K = 4$ communities: $\pi_{\text{low}}$ (left) and $\pi_{\text{high}}$ (right).}
  \label{fig:ELBO_BDSBM}
\end{figure*}

Figure~\ref{fig:memberships_BDSBM} shows the posterior community membership probabilities for all individuals in the two data sets. 
For both regimes, the posterior probabilities are typically very close to $0$ or $1$, especially in the high-signal regime, indicating that individuals are assigned to a single community with high confidence. 
Most exceptions involve the most recent individuals (those arriving later in the process), particularly in the low-signal case, where fewer interactions are available.
Overall, the posterior distributions are very similar across the two regimes, with sharper assignments under $\pi_{\text{high}}$.

\begin{figure*}[h!]
  \centering
  \begin{subfigure}[b]{.45\textwidth}
    \centering
    \includegraphics[width=0.8\linewidth]{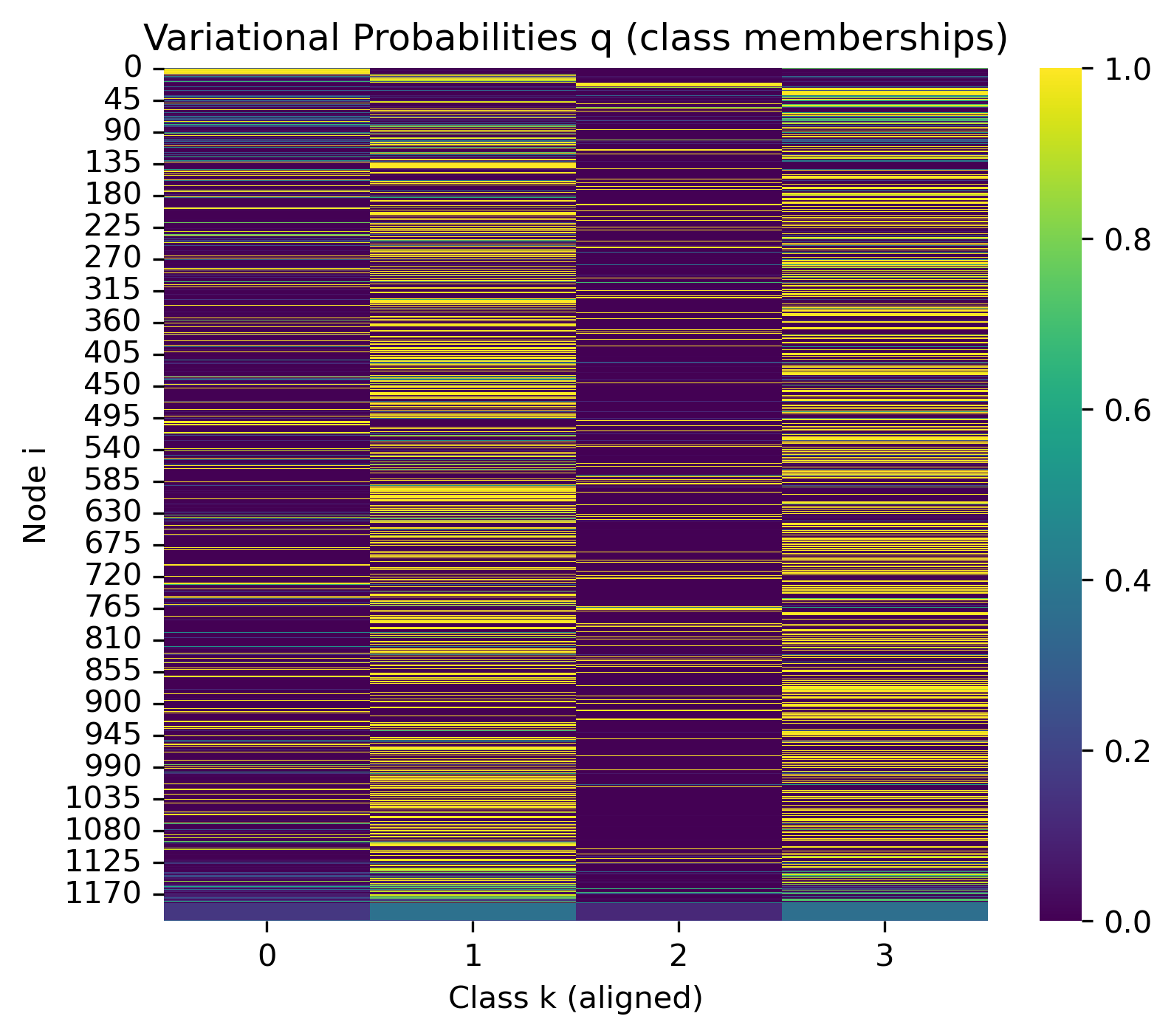}
    \caption{$\pi_{\text{low}}$.}
    \label{fig:memberships_BDSBM_pi_small}
  \end{subfigure}
  \hfill
  \begin{subfigure}[b]{.45\textwidth}
    \centering
    \includegraphics[width=0.8\linewidth]{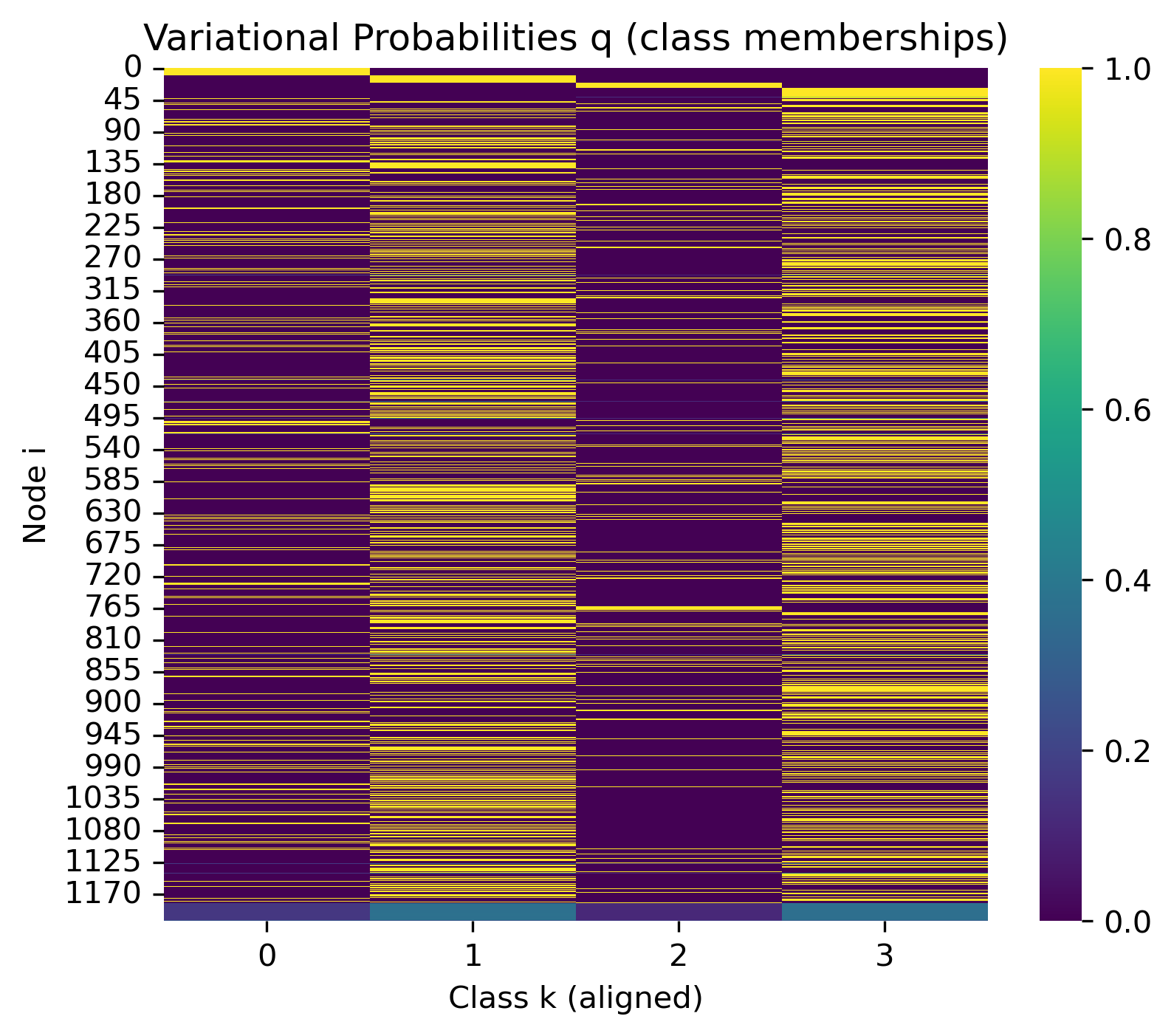}
    \caption{$\pi_{\text{high}}$.}
    \label{fig:memberships_BDSBM_pi_big}
  \end{subfigure}
  \caption{Posterior community membership probabilities for the two simulated BD-SBM networks with $K = 4$ communities. Left: data generated with $\pi_{\text{low}}$; right: data generated with $\pi_{\text{high}}$.}
  \label{fig:memberships_BDSBM}
\end{figure*}

Figure~\ref{fig:accuracy_BDSBM} compares the true and inferred community memberships. 
The confusion matrices in Figures~\ref{fig:confusion_matrix_BDSBM_pi_small} and~\ref{fig:confusion_matrix_BDSBM_pi_big} show that, in both scenarii, the vast majority of individuals are assigned to their correct communities, with only limited off-diagonal elements.

\begin{figure*}[h!]
  \centering
  \begin{subfigure}[b]{.4\textwidth}
    \centering
    \includegraphics[width=\linewidth]{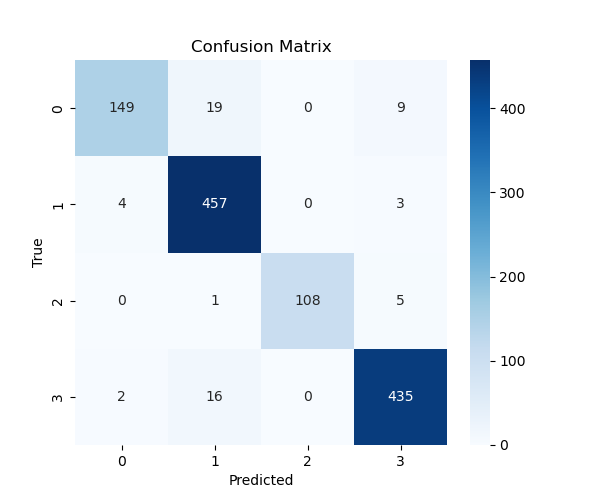} 
    \captionsetup{width=0.8\textwidth}
    \caption{Confusion matrix for the data generated with $\pi_{\text{low}}$.}
    \label{fig:confusion_matrix_BDSBM_pi_small}
  \end{subfigure}
  \hfill
  \begin{subfigure}[b]{.4\textwidth}
    \centering
    \includegraphics[width=\linewidth]{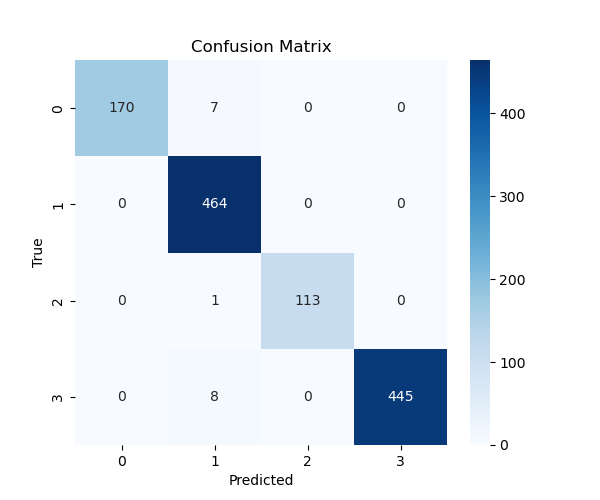} 
   \captionsetup{width=0.8\textwidth}
    \caption{Confusion matrix for the data generated with $\pi_{\text{high}}$.}
    \label{fig:confusion_matrix_BDSBM_pi_big}
  \end{subfigure}
  \\
  \begin{subfigure}[b]{.4\textwidth}
    \centering
    \includegraphics[width=\linewidth]{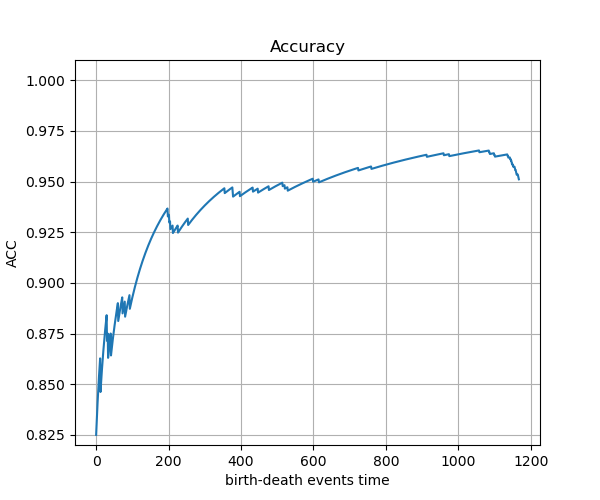}
    \captionsetup{width=0.8\textwidth}
    \caption{Evolution of clustering accuracy over time for $\pi_{\text{low}}$.}
    \label{fig:accuracy_on_time_BDSBM_pi_small}
  \end{subfigure}
  \hfill
  \begin{subfigure}[b]{.4\textwidth}
    \centering
    \includegraphics[width=\linewidth]{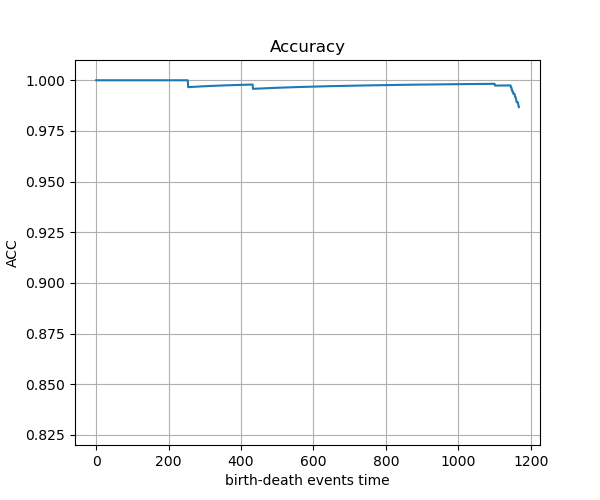}
    \captionsetup{width=0.8\textwidth}
    \caption{Evolution of clustering accuracy over time for $\pi_{\text{high}}$.}
    \label{fig:accuracy_on_time_BDSBM_pi_big}
  \end{subfigure}
  \caption{Comparison between true and inferred communities for the two simulated BD-SBM networks with $K = 4$ communities. Top: confusion matrices. Bottom: evolution of clustering accuracy over time as the population evolves.}
  \label{fig:accuracy_BDSBM}
\end{figure*}

Figures~\ref{fig:accuracy_on_time_BDSBM_pi_small} and~\ref{fig:accuracy_on_time_BDSBM_pi_big} display the evolution of clustering accuracy as a function of time (i.e., as the population size evolves). 
For the low-signal case $\pi_{\text{low}}$, accuracy is relatively unstable at early times, then increases and remains high during the middle part of the process before decreasing slightly at the very end of the observation window. 
This behaviour is expected: information about individuals present at $t_0$ is limited, whereas for individuals arriving later, both their entry time and the population size at entry are known, and their subsequent interactions accumulate over time.
Examining accuracy as a function of the proportion of initial individuals $N_0$ relative to the total population, we find that accuracy increases as this ratio decreases. In the high-signal case $\pi_{\text{high}}$, accuracy rapidly reaches a high level and remains largely stable, with only a slight decline for individuals arriving at the very end of the observation window, those who have very few recorded interactions. As expected, overall performance is superior in the high-signal scenario, where intra-community connection probabilities clearly outweigh inter-community ones.

\medskip
The birth and death rates $\lambda$ and $\mu$ do not depend on the latent communities, and their estimates are very close to the true values. 
Table~\ref{tab:params_bd_true_est} reports both the true and the estimated values of the birth and death rates.

\begin{table}[ht]
    \centering
    \captionsetup{width=\textwidth}
    \begin{tabular}{ccc}
        \toprule
        Parameter & True value & Estimated value \\
        \midrule
        $\lambda$ & $0.04$ & $0.0400$ \\
        $\mu$     & $0.02$ & $0.0207$ \\
        \bottomrule
    \end{tabular}
    \caption{True vs. estimated birth-death parameter values.}
    \label{tab:params_bd_true_est}
\end{table}

The estimation of the remaining model parameters is summarized in Table~\ref{tab:params_true_est} and Figure~\ref{fig:estimated_pi_BDSBM}. For both regimes $\pi_{\text{low}}$ and $\pi_{\text{high}}$,  the estimates of the community probabilities at $t_0$, i.e. $\beta$, are reported in Table~\ref{tab:params_true_est}. 
In the low-signal case, this estimation is slightly less accurate than in the high-signal case, which is consistent with the lower clustering performance at early times. 

\begin{figure*}[h!]
  \centering
  \begin{subfigure}[b]{.9\textwidth}
    \centering
    \includegraphics[width=\textwidth]{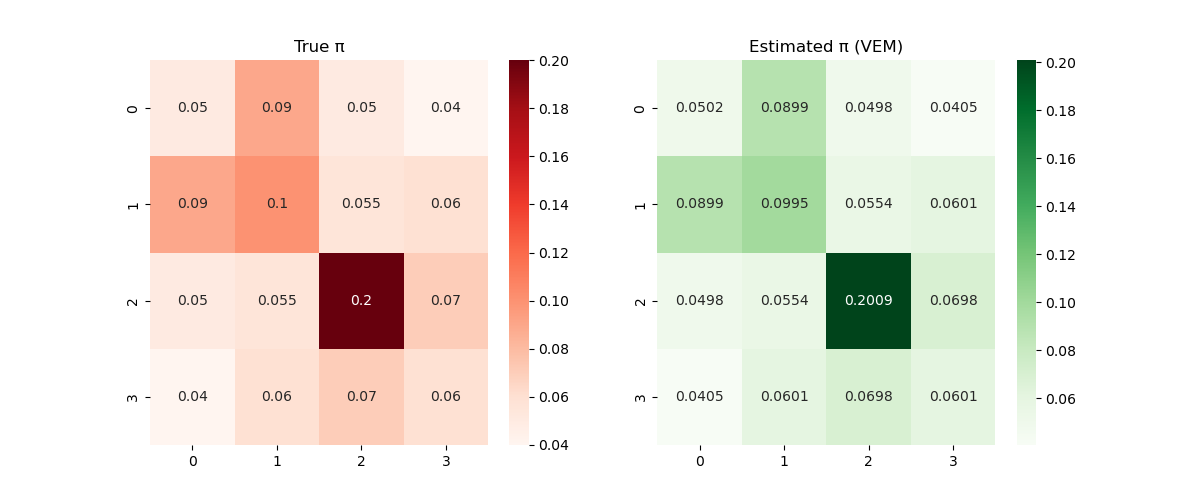} 
    \caption{Data generated with $\pi_{\text{low}}$.}
    \label{fig_7a}
  \end{subfigure}
  \\
  \begin{subfigure}[b]{.9\textwidth}
    \centering
    \includegraphics[width=\textwidth]{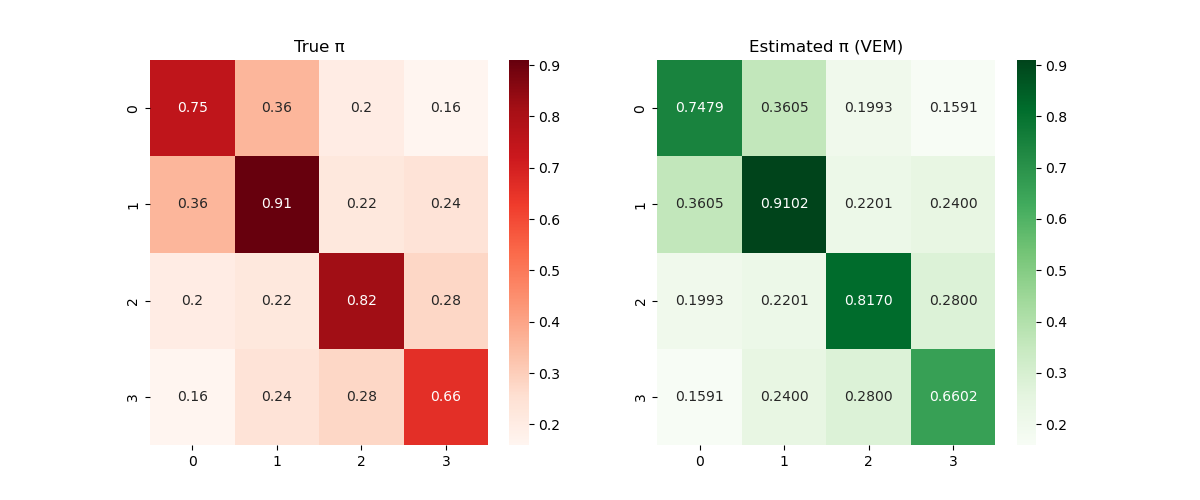} 
    \caption{Data generated with $\pi_{\text{high}}$.}
    \label{fig_7b}
  \end{subfigure}
  \caption{Estimated block probability matrices $\widehat{\pi}$ (right of each panel) compared to the true matrices $\pi$ (left) for the two BD-SBM scenarios.}
  \label{fig:estimated_pi_BDSBM}
\end{figure*}

\begin{table}[ht]
    \centering
    \captionsetup{width=\textwidth}
    \begin{tabular}{cccc}
        \toprule
        Parameter & True value & Estimated ($\pi_{\text{low}}$) & Estimated ($\pi_{\text{high}}$) \\
        \midrule
        $\beta_0$ & 0.25  & 0.2751 & 0.2525 \\
        $\beta_1$ & 0.275 & 0.2463 & 0.2725 \\
        $\beta_2$ & 0.175 & 0.1622 & 0.1750 \\
        $\beta_3$ & 0.30  & 0.3165 & 0.3000 \\
        \bottomrule
    \end{tabular}
    \caption{True vs. estimated community proportions $\beta$ for the BD-SBM.}
    \label{tab:params_true_est}
\end{table}

The estimated block probability matrices $\widehat{\pi}$ in Figure~\ref{fig:estimated_pi_BDSBM} closely reproduce the true matrices $\pi_{\text{low}}$ and $\pi_{\text{high}}$ used in the simulations. 
While minor deviations appear in blocks with low connection probabilities, the relative ordering of intra- and inter-community connection strengths is preserved, indicating that the model effectively captures both structural heterogeneity and the impact of temporal population turnover.

\medskip
Allowing for departures in our model makes the inference problem both statistically and algorithmically more challenging. In particular, when individuals are allowed to leave the network, the process governing community sizes is no longer Markovian at departure times, rendering the exact VEM updates intractable.
This motivates the approximation introduced in Section~\ref{sec:vem_inf}, which restores tractability at the cost of a slight weakening  of the theoretical guarantees on ELBO monotonicity.

\subsubsection{Pure birth  process within dSBM}\label{BSBM_experiments}

We now turn to the pure birth process within a dynamic SBM, where the population size may only increase over time ($\mu = 0$), while the within- and between-community connection probabilities are the same block matrices $\pi_{\text{low}}$ and $\pi_{\text{high}}$ defined in \ref{sec:simu}. 
This scenario serves as a baseline in which temporal dynamics arise exclusively from arrivals, allowing us to isolate the impact of departures on community recovery.

As for the BD-SBM case, we first choose the optimal $K$ using  the ICL criterion described in section \ref{sec:ICL}.
In particular, we perform $50$ random initializations of the VEM algorithm and, for each of them, we retain the value of $K$ that maximizes the  ICL. Figure~\ref{fig:ICL_BSBM} reports, separately for $\pi_{\text{low}}$ and $\pi_{\text{high}}$, how often each value of $K$ is selected across the random runs. In both regimes, the ICL criterion most frequently selects $K = 4$, in agreement with the true number of communities used to generate the data.

\begin{figure*}[h!]
  \centering
  \begin{subfigure}[b]{.4\textwidth}
    \centering
    \includegraphics[width=0.8\linewidth]{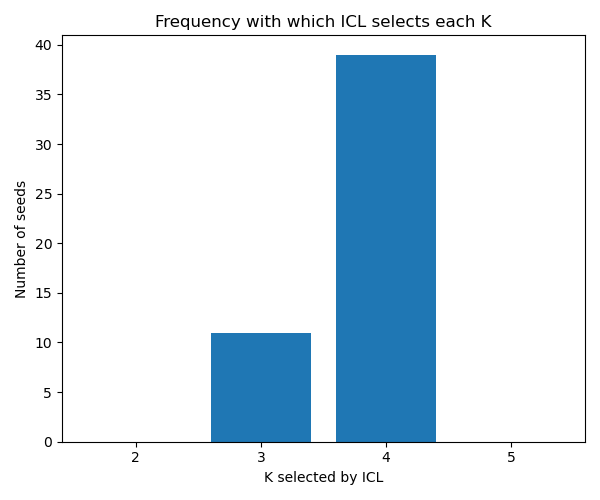}
    \caption{$\pi_{\text{low}}$.}
    \label{fig:ICL_BSBM_pi_small}
  \end{subfigure}
  \hfill
  \begin{subfigure}[b]{.4\textwidth}
    \centering
    \includegraphics[width=0.8\linewidth]{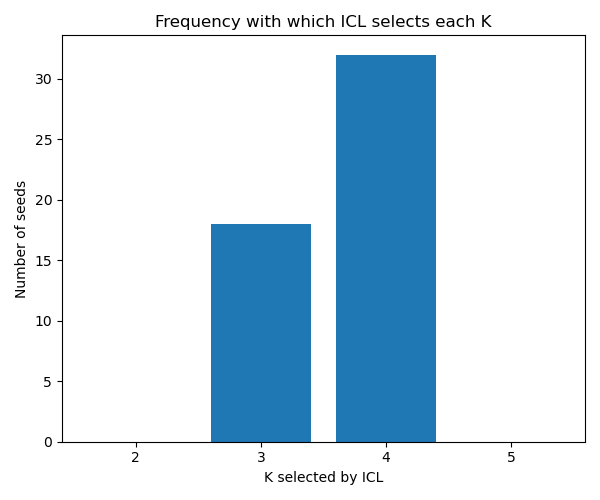} 
    \caption{$\pi_{\text{high}}$.}
    \label{fig:ICL_BSBM_pi_big}
  \end{subfigure}
  \caption{Number of times the ICL criterion selects each value of $K$ as the optimal number of communities over all random seeds in the BSBM experiment with $\pi_{\text{low}}$ (left) and $\pi_{\text{high}}$ (right).}
  \label{fig:ICL_BSBM}
\end{figure*}

We therefore fix $K = 4$ and apply the proposed VEM algorithm using the same initialization procedure as in the BD-SBM experiment (described in Subsection~\ref{initialization}). In this pure birth setting, the algorithm converges in a small number of iterations. Figure~\ref{fig:ELBO_BSBM} displays the evolution of the ELBO over the VEM iterations for both connectivity regimes. As in the BD-SBM case, the ELBO trajectories quickly stabilize, indicating that the algorithm reaches a local optimum after only a few updates.

\begin{figure*}[h!]
  \centering
  \begin{subfigure}[b]{.45\textwidth}
    \centering
    \includegraphics[width=0.8\linewidth]{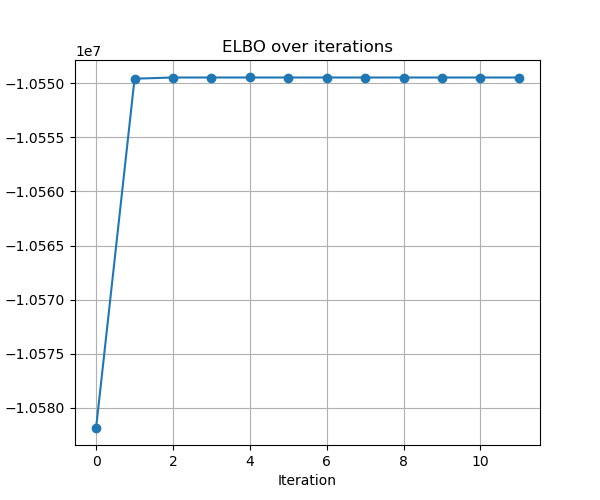}
    \caption{$\pi_{\text{low}}$.}
    \label{fig:ELBO_BSBM_pi_small}
  \end{subfigure}
  \hfill
  \begin{subfigure}[b]{.45\textwidth}
    \centering
    \includegraphics[width=0.8\linewidth]{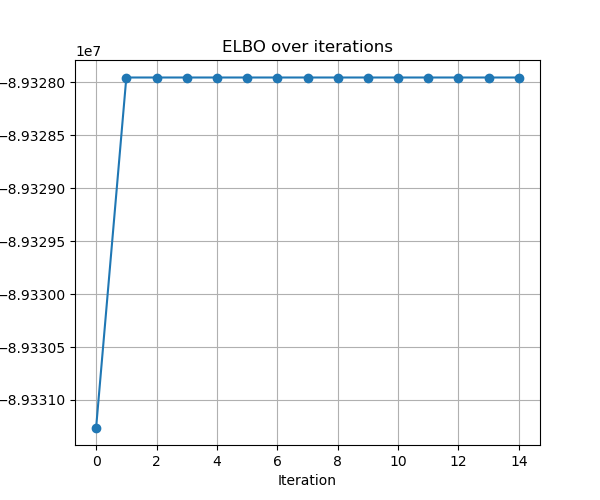} 
    \caption{$\pi_{\text{high}}$.}
    \label{fig:ELBO_BSBM_pi_big}
  \end{subfigure}
  \caption{Evolution of the ELBO over VEM iterations for the two simulated B-SBM networks with $K = 4$ communities: $\pi_{\text{low}}$ (left) and $\pi_{\text{high}}$ (right).}
  \label{fig:ELBO_BSBM}
\end{figure*}

Figure~\ref{fig:memberships_BSBM} displays the posterior community membership probabilities for all individuals in the two B-SBM data sets. We observe the same phenomena as for the BD-SBM: 
the posterior probabilities are typically very close to $0$ or $1$, indicating highly confident assignments. 
The absence of departures makes the trajectory of each individual longer on average, which in turn yields more interaction data and thus facilitates community identification.

\begin{figure*}[h!]
  \centering
  \begin{subfigure}[b]{.4\textwidth}
    \centering
    \includegraphics[width=\linewidth]{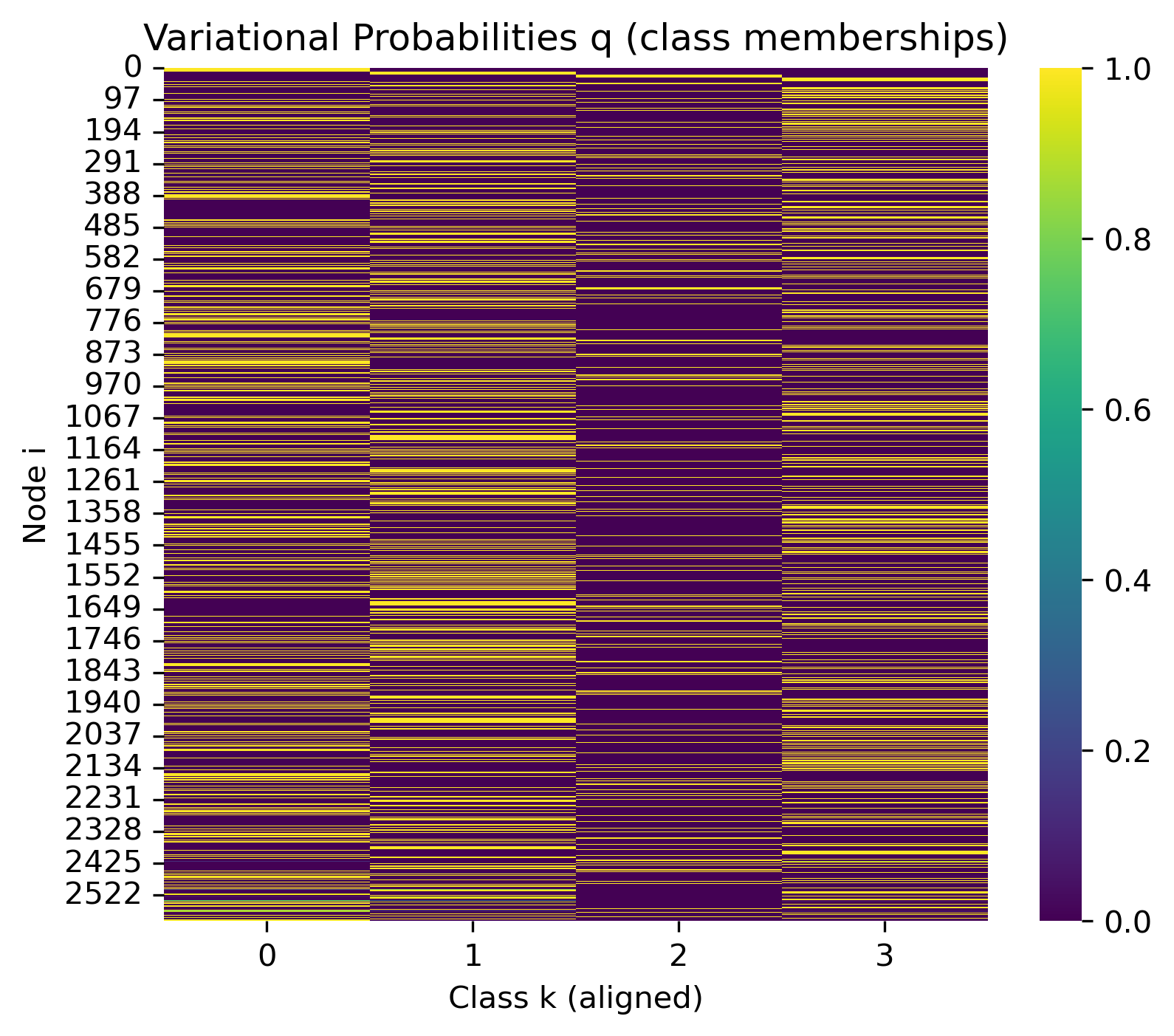}
    \caption{$\pi_{\text{low}}$.}
    \label{fig:memberships_BSBM_pi_small}
  \end{subfigure}
  \hfill
  \begin{subfigure}[b]{.4\textwidth}
    \centering
    \includegraphics[width=\linewidth]{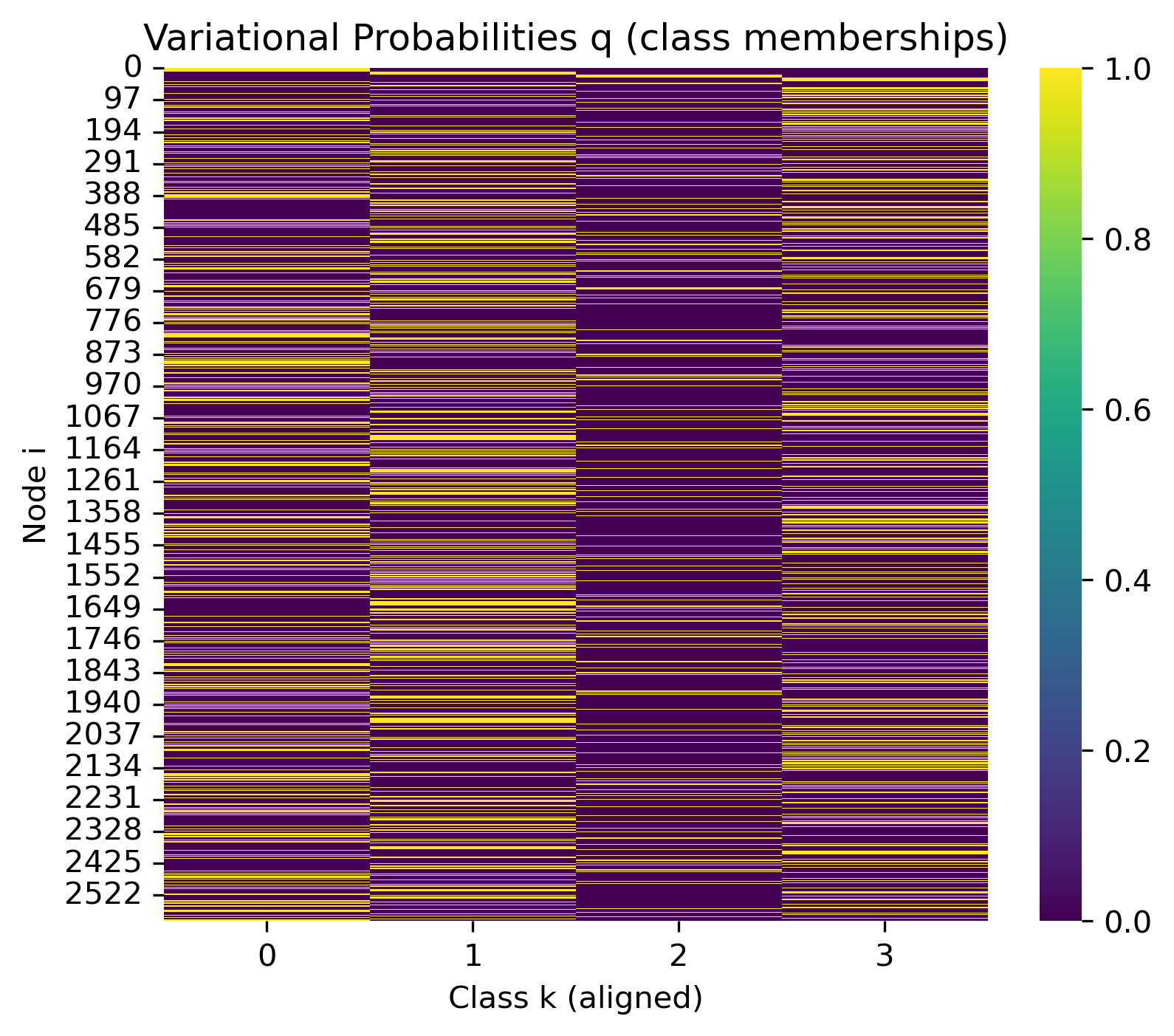}
    \caption{$\pi_{\text{high}}$.}
    \label{fig:memberships_BSBM_pi_big}
  \end{subfigure}
  \caption{Posterior community membership probabilities for the two simulated BSBM networks with $K = 4$ communities. Left: data generated with $\pi_{\text{low}}$; right: data generated with $\pi_{\text{high}}$.}
  \label{fig:memberships_BSBM}
\end{figure*}

We next compare the inferred community labels to the true ones. 
Figure~\ref{fig:accuracy_BSBM} reports the confusion matrices for $\pi_{\text{low}}$ and $\pi_{\text{high}}$. 
In both signal regimes, the model achieves perfect clustering accuracy: off-diagonal entries in the confusion matrices are essentially zero, and the accuracy remains equal to one throughout the observation window. 
This sharp contrast with the BD-SBM results highlights the additional difficulty introduced by departures, which both require an approximation in the VEM algorithm and shorten some trajectories, thereby reducing the amount of information available for those individuals.

\begin{figure*}[h!]
  \centering
  \begin{subfigure}[b]{.4\textwidth}
    \centering
    \includegraphics[width=\linewidth]{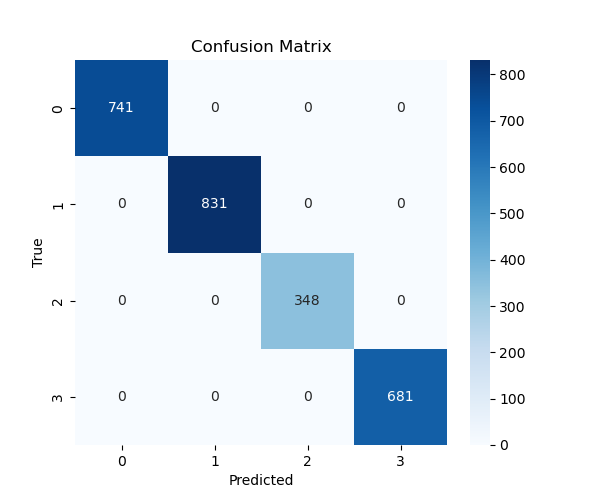} 
     \captionsetup{width=0.8\textwidth}
    \caption{Confusion matrix for the data generated with $\pi_{\text{low}}$.}
    \label{fig:confusion_matrix_BSBM_pi_small}
  \end{subfigure}
  \hfill
  \begin{subfigure}[b]{.4\textwidth}
    \centering
    \includegraphics[width=\linewidth]{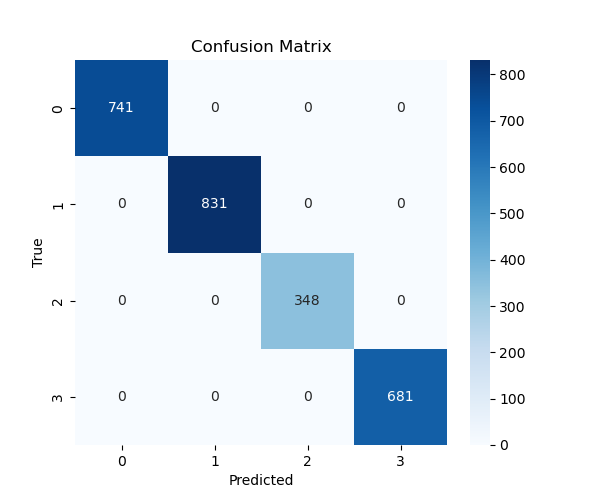} 
     \captionsetup{width=0.8\textwidth}
    \caption{Confusion matrix for the data generated with $\pi_{\text{high}}$.}
    \label{fig:confusion_matrix_BSBM_pi_big}
  \end{subfigure}
  \caption{Comparison between true and inferred communities for the two simulated BSBM networks with $K = 4$ communities. Top: confusion matrices. Bottom: evolution of clustering accuracy over time.}
  \label{fig:accuracy_BSBM}
\end{figure*}

\medskip
Finally, we examine parameter estimation in the pure birth setting. 
Table~\ref{tab:params_b_true_est} reports the true and estimated values of the birth rate \(\lambda\).
The estimated value is extremely close to the true one, confirming that the model accurately recovers the underlying birth dynamics when no departures are present. 

\medskip
The estimated block probability matrices \(\widehat{\pi}\) in Figure~\ref{fig:estimated_pi_BSBM} also closely match the true matrices \(\pi_{\text{low}}\) and \(\pi_{\text{high}}\), preserving the ordering of within- and between-community connection strengths in both regimes.

\begin{figure*}[h!]
  \centering
  \begin{subfigure}[b]{.9\textwidth}
    \centering
    \includegraphics[width=\textwidth]{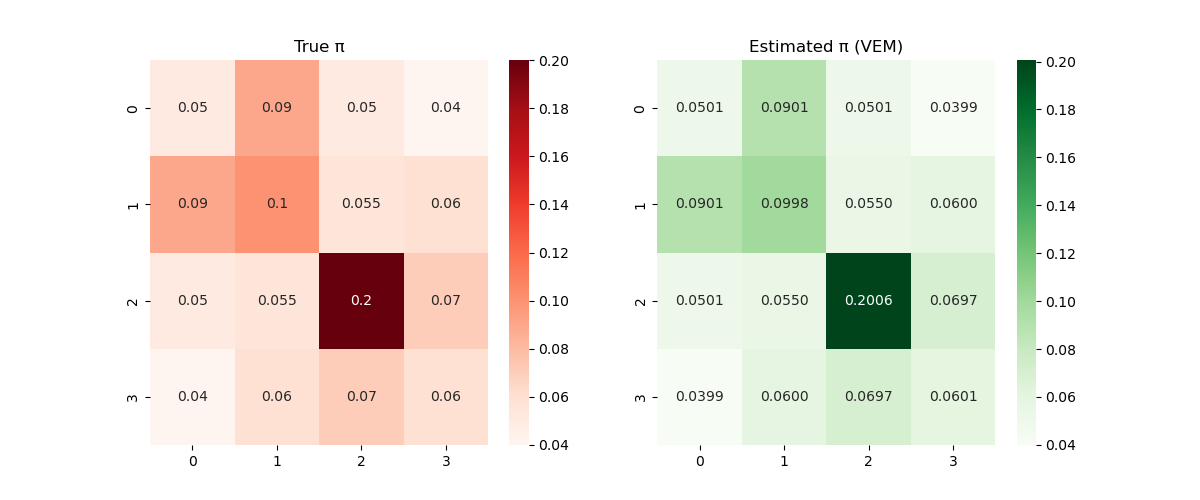} 
     \captionsetup{width=0.8\textwidth}
    \caption{Data generated with $\pi_{\text{low}}$.}
    \label{fig:estimated_pi_BSBM_pi_small}
  \end{subfigure}
  \\
  \begin{subfigure}[b]{.9\textwidth}
    \centering
    \includegraphics[width=\textwidth]{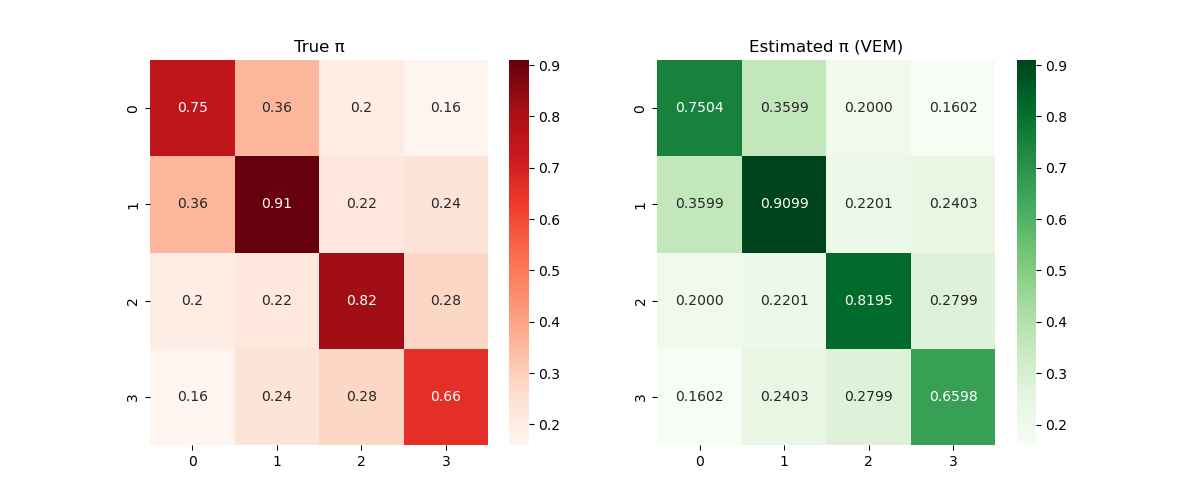} 
     \captionsetup{width=0.8\textwidth}
    \caption{Data generated with $\pi_{\text{high}}$.}
    \label{fig:estimated_pi_BSBM_pi_big}
  \end{subfigure}
  \caption{Estimated block probability matrices $\widehat{\pi}$ (right of each panel) compared to the true matrices $\pi$ (left) for the two BSBM scenarios.}
  \label{fig:estimated_pi_BSBM}
\end{figure*}

\begin{table}[ht]
    \centering
    \captionsetup{width=\textwidth}
    \begin{tabular}{cccc}
        \toprule
        Parameter & True value & Estimated value  \\
        \midrule
        $\lambda$ & 0.04 & 0.0405\\
        \bottomrule
    \end{tabular}
     \captionsetup{width=0.8\textwidth}
    \caption{True vs. estimated birth-death parameter values.}
    \label{tab:params_b_true_est}
\end{table}

\newpage

\subsection{Data Extraction from  \texttt{arXiv}}
The dataset under study is constructed from the \texttt{arXiv} electronic repository. 
We focus on all articles  categorized under  \texttt{Mathematics} and published between January~2019 and September~2025. Metadata are retrieved through the official \texttt{arXiv API}, which provides structured access to publication records, including titles, author lists, subject categories, and submission dates. 

Each article is treated as an undirected temporal interaction among its co-authors (distinct nodes) and makes the pair of co-authors connected by an edge dated at the time of the paper’s first submission (the \texttt{published} field). In addition, to capture revision activity, we treat different versions of the same paper as distinct temporal interaction events between the same set of co-authors: version~1 is timestamped at \texttt{published}, while subsequent versions (if any) are timestamped at \texttt{updated}. The resulting temporal co-authorship network spans more than five years of activity and covers all mathematical subfields in \texttt{arXiv} (e.g., \texttt{math.ST}, \texttt{math.PR}, \texttt{math.OC}, etc., see table \ref{tab:arxiv_math_categories}). To summarize, nodes represent individual authors while temporal edges represent collaboration events occurring at specific dates.

The raw scale of the \texttt{Mathematics} corpus over this period exceeds \num{200000} works and \num{50000} unique authors. Because daily-resolution snapshots are extremely sparse and full-graph inference would be computationally prohibitive, we first restrict attention to a denser working network of size $N=5000$ authors, chosen to maximize edge density under this size constraint.

From this reduced dataset, we define the ancestor set $V_0$ as the authors who published at least once during the first six months (182 days) of 2019. The observation window for our analysis begins in July~2019. Starting from $V_0$, we iteratively expand the node set by adding all authors who co-author with any member of the current set at least once within the observation window, and we repeat this expansion until no new authors are added. The transitive co-authorship neighborhood grown from $V_0$, together with all co-authorship links among its members, forms the induced subgraph we analyze. Applying this construction yields $|V_0|=2365$ ancestors and a total of $N=4622$ authors (ancestors plus their retained collaborators) present in the sequence of interaction networks during the observation window.  The observed total number of birth–death events is $M = 3725$.

To mitigate sparsity while retaining temporal structure, we discretize time into two-month bins, resulting in $38$ time steps over the study period. Within each bin, two authors are linked if they co-author at least one paper (or version-event) during that bin. For authors who do not belong to $V_0$ (i.e., non-ancestors), we define their birth time as the date of their first observed publication, and their death time as $\min(\text{last-publication-date} + \text{2 months},\, t_T)$, where $t_T$ denotes the end of the observation window. In the rare case of identical birth or death times for multiple authors, we break ties by discarding the author(s) with fewer publications, ensuring unique event times. 

\begin{figure*}[h!]
  \centering
  \begin{subfigure}[b]{.45\textwidth}
    \centering
    \includegraphics[width=0.8\linewidth]{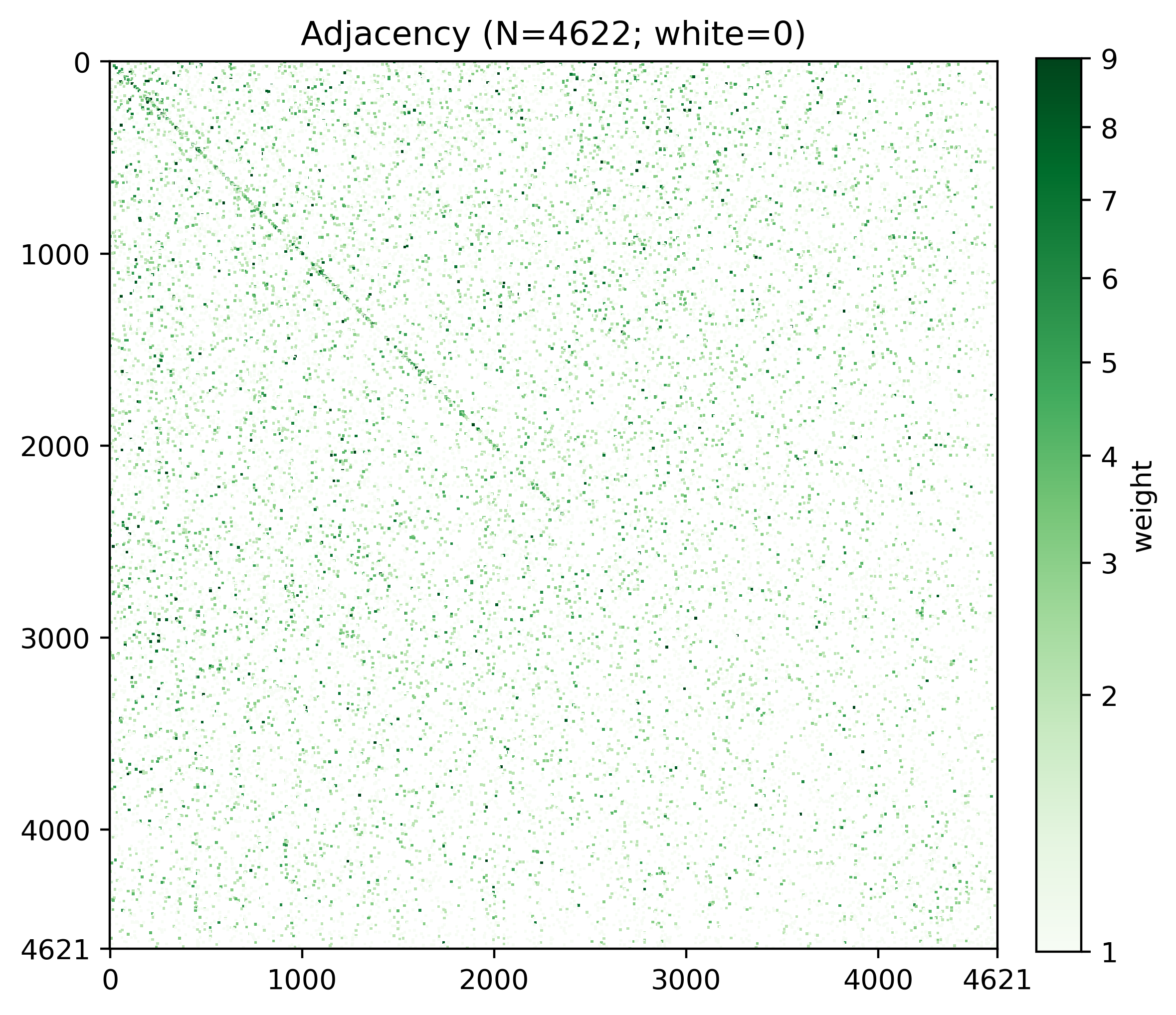} 
     \captionsetup{width=0.8\textwidth}
    \caption{Adjacency matrix of the working arXiv co-authorship network.}
    \label{fig:adjacency_matrix_arXiv_a}
  \end{subfigure}
  \hfill
  \begin{subfigure}[b]{.45\textwidth}
    \centering
    \includegraphics[width=0.8\linewidth]{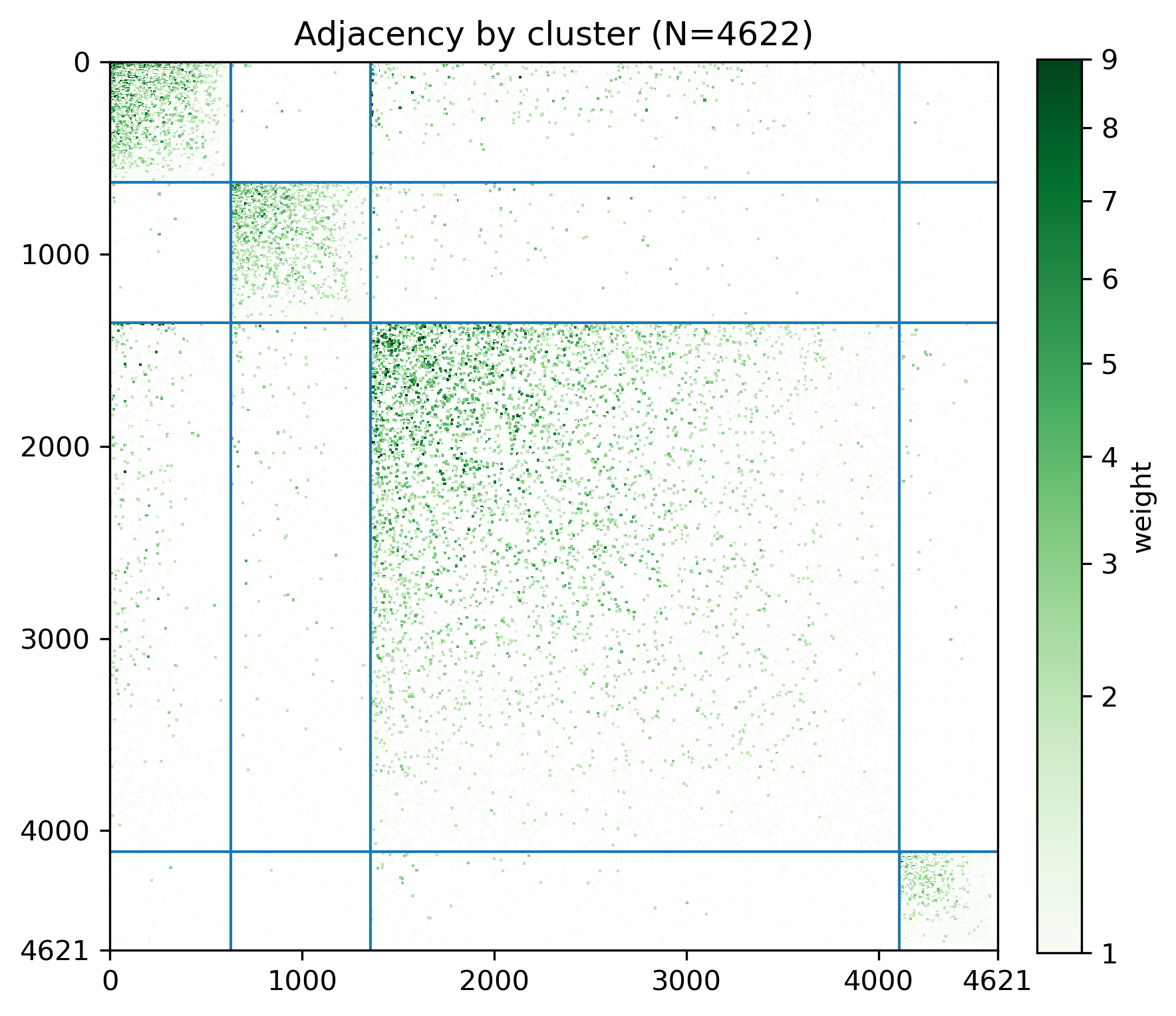}
     \captionsetup{width=0.8\textwidth}
    \caption{Adjacency matrix after reordering authors by inferred community.}
    \label{fig:adjacency_matrix_arXiv_b}
  \end{subfigure}
  \caption{Adjacency structure of the arXiv Mathematics co-authorship network.}
  \label{fig:adjacency_matrix_arXiv}
\end{figure*}

\medskip
We fit our birth–death SBM model (BD-SBM) to this temporally aggregated co-authorship network.

\textit{Model selection.} First, we perform model selection over the number of communities by running our VEM algorithm for $K \in \{2,3,4\}$ and computing the Integrated Completed Likelihood (ICL) for each candidate. Given the large product $N \times M$, we deliberately restrict attention to small values of $K$ to keep computation tractable while still allowing for non-trivial community structure. The best trade-off between fit and complexity is obtained for $K = 4$, which attains the highest ICL and is therefore retained in the subsequent analysis.

\textit{Initialization of the VEM algorithm.} As mentioned in section \ref{initialization} the initialization of a VEM algorithm plays a crucial role. Given  the observation of a very sparse dynamic networks, 
we therefore use a dynamic extension of the sparse SBM proposed in \cite{frisch2021sparsebm} to construct  the initialization for the latent memberships. 
In practice, we first fit this dynamic sparse SBM to obtain coarse cluster assignments and then use the resulting soft memberships as initial values for our BD-SBM VEM algorithm.

\textit{Post-inference results.} Figure~\ref{fig:adjacency_matrix_arXiv} summarizes the global connectivity structure of the \texttt{arXiv} co-authorship network, showing the raw adjacency matrix in Figure~\ref{fig:adjacency_matrix_arXiv_a} while displaying in  Figure~\ref{fig:adjacency_matrix_arXiv_b} the same matrix after permuting authors according to their inferred community labels. This highlights an overall low level of interaction among authors, with stronger interaction intensity observed within communities than between them.

The inference results are summarized in Figures~\ref{fig:inference_arxiv_data} and~\ref{fig:meen_rate_cononection}. 
Figure~\ref{fig:inference_arxiv_data_a} shows the posterior community membership probabilities for all authors, while Figure~\ref{fig:inference_arxiv_data_b} reports the estimated block probability matrix $\widehat{\pi}$. 
The ELBO trajectory in Figure~\ref{fig:inference_arxiv_data_c} indicates that the algorithm converges and stabilizes after approximately 25 iterations.

\begin{figure*}[h!]
  \centering
  \begin{subfigure}[b]{.3\textwidth}
    \centering
    \includegraphics[width=\linewidth]{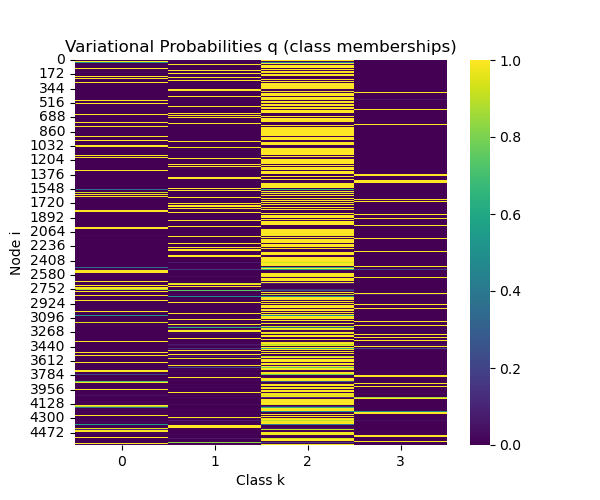}
     \captionsetup{width=0.8\textwidth}
    \caption{Posterior community membership.}
    \label{fig:inference_arxiv_data_a}
  \end{subfigure}
  \hfill
  \begin{subfigure}[b]{.3\textwidth}
    \centering
    \includegraphics[width=\linewidth]{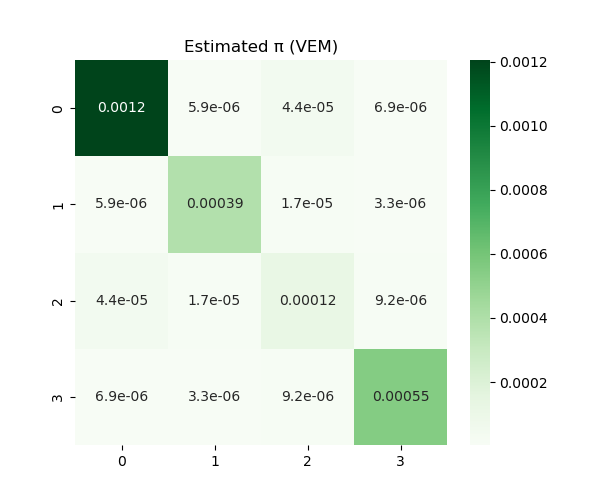} 
     \captionsetup{width=0.8\textwidth}
    \caption{Estimated block probability matrix $\widehat{\pi}$.}
    \label{fig:inference_arxiv_data_b}
  \end{subfigure}
  \hfill
  \begin{subfigure}[b]{.3\textwidth}
    \centering
    \includegraphics[width=\linewidth]{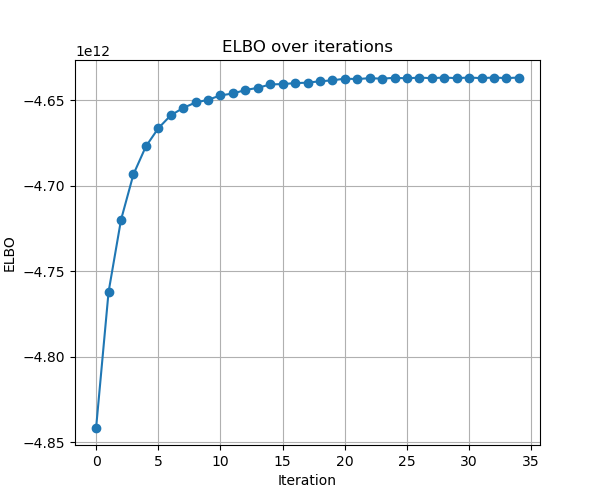} 
     \captionsetup{width=0.8\textwidth}
    \caption{ELBO over VEM iterations.}
    \label{fig:inference_arxiv_data_c}
  \end{subfigure}
  \caption{Inference results for the arXiv Mathematics co-authorship network.}
  \label{fig:inference_arxiv_data}
\end{figure*}

To better understand the structural role of the inferred communities, we analyse mean connection rates by cluster (Figure~\ref{fig:meen_rate_cononection}). 
Figure~\ref{fig:meen_rate_cononection_a} reports the average degree (mean connection rate) of authors in each cluster, while Figure~\ref{fig:meen_rate_cononection_b} focuses on the mean rate of within-cluster connections. 
Cluster~2 exhibits a low within-cluster rate, suggesting a weakly connected community compared to the others.

\begin{figure*}[h!]
  \centering
  \begin{subfigure}[b]{.45\textwidth}
    \centering
    \includegraphics[width=\linewidth]{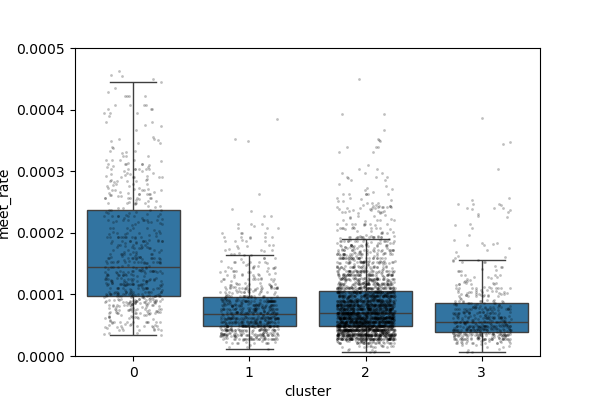}
     \captionsetup{width=0.8\textwidth}
    \caption{Individual connection rate (average degree) by cluster.}
    \label{fig:meen_rate_cononection_a}
  \end{subfigure}
  \hfill
  \begin{subfigure}[b]{.45\textwidth}
    \centering
    \includegraphics[width=\linewidth]{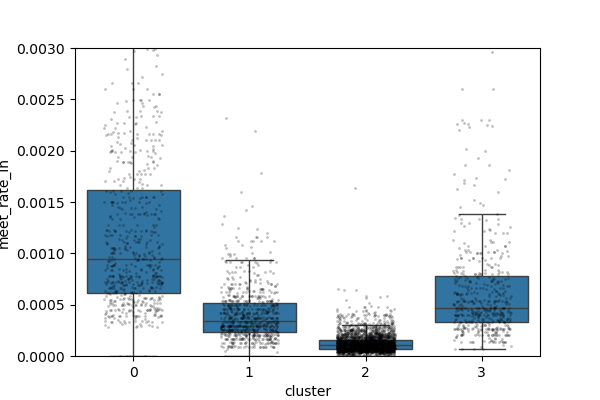} 
    \caption{Individual within-cluster connection rate by cluster.}
    \label{fig:meen_rate_cononection_b}
  \end{subfigure}
   \captionsetup{width=0.8\textwidth}
  \caption{Boxplots of individual connection rates by inferred communities}
  \label{fig:meen_rate_cononection}
\end{figure*}

Figure~\ref{fig:prop_by_cluster} provides a content-based interpretation of the inferred communities. 
For each author, we consider as ``primary category'' every \texttt{arXiv} primary subject in which they have at least one paper during the study period (an author can therefore contribute to several categories). 
We retain the top $20$ categories with the largest number of authors. For each category, we plot its proportion within each cluster.
We observe that clusters differ not only in connectivity patterns but also in their predominant mathematical areas, indicating that the inferred communities capture both structural and thematic heterogeneity. 

\begin{figure}[h!]
  \centering
  \includegraphics[width=.99\textwidth]{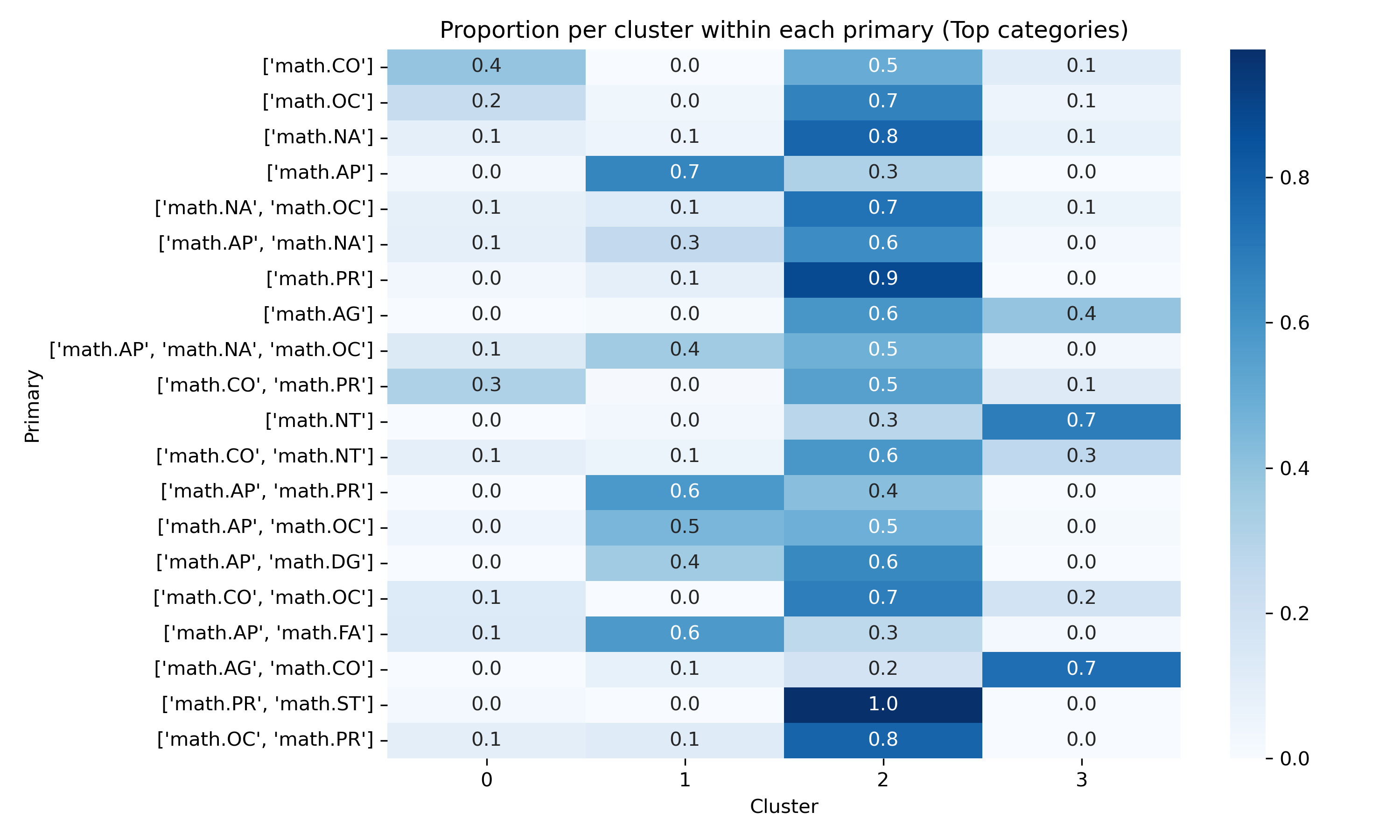}
   \captionsetup{width=0.8\textwidth}
  \caption{Top \texttt{arXiv} Mathematics subject categories by inferred community (proportion within each cluster).}
  \label{fig:prop_by_cluster}
\end{figure}

Cluster~2 gathers many authors who tend to collaborate less frequently (those with lower mean degree). 
This cluster is particularly enriched in Optimization and Control (\texttt{math.OC}), Numerical Analysis (\texttt{math.NA}), and authors whose primary labels combine these areas, such as [\texttt{math.OC}, \texttt{math.NA}]. 
Authors publishing in Probability (\texttt{math.PR}), or jointly in [\texttt{math.PR}, \texttt{math.ST}] and [\texttt{math.PR}, \texttt{math.OC}], also tend to fall into this low-collaboration cluster. 
Cluster~0 is dominated by Combinatorics (\texttt{math.CO}). 
Cluster~1 is strongly influenced by Analysis of PDEs (\texttt{math.AP}); it contains both authors who publish exclusively in \texttt{math.AP} and those whose primary labels combine \texttt{math.AP} with other areas, such as [\texttt{math.AP}, \texttt{math.PR}], [\texttt{math.AP}, \texttt{math.OC}] or [\texttt{math.AP}, \texttt{math.DG}]. 
Cluster~3 includes a large fraction of authors in Number Theory (\texttt{math.NT}), Algebraic Geometry (\texttt{math.AG}) and their combination, e.g.\ [\texttt{math.AG}, \texttt{math.CO}]. 
See Table~\ref{tab:arxiv_math_categories} for details on the subject-category labels.

Finally, Table~\ref{tab:params_arxiv_est} reports the estimated birth and death rates, together with the inferred community proportions $\beta_k$ for $K=4$. 
The fitted model suggests a relatively small but non-zero departure rate, together with a dominant community (cluster~2) containing approximately $60\%$ of the authors.

\begin{table}[ht]
    \centering
    \captionsetup{width=\textwidth}
    \begin{tabular}{cc}
        \toprule
        Parameter & Estimated value \\
        \midrule
        $\lambda$  & $0.271 \cdot 10^{-3}$ \\
        $\mu$      & $0.176  \cdot  10^{-3}$ \\
        $\beta_0$  & $0.1420$ \\
        $\beta_1$  & $0.1467$ \\
        $\beta_2$  & $0.6012$ \\
        $\beta_3$  & $0.1101$ \\
        \bottomrule
    \end{tabular}
    \caption{Estimated BD-SBM parameters for the \texttt{arXiv} Mathematics co-authorship network.}
    \label{tab:params_arxiv_est}
\end{table}

The clustering obtained from the BD-SBM reveals marked heterogeneity among authors: it highlights sparsely connected nodes, researchers who tend to collaborate less within all communities, and groups of authors whose collaboration patterns align closely with the primary \texttt{arXiv} subject categories. Taken together, these results illustrate the ability of the BD-SBM to disentangle latent research communities while accounting for the demographic turnover induced by the arrival and departure of authors over time.

\section{Conclusion}\label{sec:conclusion}

In this work we have introduced the birth-death stochastic block model (BD-SBM), a dynamic stochastic block model that explicitly accounts for demographic turnover. The BD-SBM combines a continuous-time birth-death process for the population size with a block-structured model for edges observed at discrete times. Individuals belong to one of $K$ latent communities, inherit the community of their parent at birth, and retain this membership throughout their lifetime. This construction allows the vertex set to evolve over time while preserving an underlying community structure that governs edge formation, and recovers classical static and dynamic SBMs with fixed vertex set as special cases.

From an inferential viewpoint, we have derived the complete-data likelihood of the BD-SBM and shown how the presence of latent community labels and unobserved community sizes naturally leads to an intractable observed-data likelihood. To overcome this difficulty, we have developed a variational expectation-maximization (VEM) algorithm tailored to the BD-SBM. Our variational family combines a mean–field approximation for individual labels present at time $t_0$ with a structured approximation for the latent community-size process, encoded through the additional variational parameters of the marginal and transition probabilities.  Within this framework, the block parameters $(\pi,\beta)$ are estimated via variational updates, while the birth-death parameters $(\lambda,\mu)$, which do not depend on the latent communities, admit closed-form maximum likelihood estimators that can be computed directly from the observed event history.

We have also proposed an adaptation of the Integrated Completed Likelihood (ICL) criterion to the BD-SBM for selecting the number of communities. The resulting penalisation reflects the complexity of the block structure through the community proportions and block connection parameters, while the birth-death parameters, which do not depend on $K$, are excluded from the penalty. maximizing this ICL criterion over a range of candidate values of $K$ provides a practical and principled tool for model selection in dynamic networks with evolving population size. In addition, we have detailed an initialization strategy for the VEM algorithm based on a similarity matrix and $K$-means clustering, which yields informative starting values for the variational distribution without enforcing the birth-death structure at initialization.

The numerical study illustrates the ability of the BD-SBM to recover latent communities and model parameters in synthetic temporal networks under two signal regimes. In a low-signal setting, where within- and between-community connection probabilities are relatively close and small (yielding a sparse network with weakly separated communities), the method still performs well and recovers the main community structure, although with lower clustering accuracy than in the high-signal case. In a high-signal regime, characterized by much stronger within-community connectivity, the BD-SBM behaves as expected and achieves substantially higher accuracy in both clustering and parameter estimation. We also considered the special case of a pure birth process, for which the algorithm coincides with an exact  VEM scheme and the community-size process is Markovian without additional approximations; in this setting, the procedure almost perfectly recovers the true parameters and labels.

We further applied the BD-SBM to a temporal co-authorship network constructed from \texttt{arXiv} mathematics articles over a range of $6$ years. Each article  has been  treated as an undirected temporal interaction among its co-authors adding an edge the pair of its co-authors at the date of the paper’s first submission. The clustering obtained from the BD-SBM reveals marked heterogeneity among authors: it highlights sparsely connected nodes, researchers who tend to collaborate less within all communities, and groups of authors whose collaboration patterns align closely with the primary \texttt{arXiv} subject categories.

Looking ahead, several directions for future research could further enrich the framework developed in this work. 
A first natural extension is to allow community-specific birth and death parameters, so that demographic dynamics may differ across communities; this extension is discussed in Appendix~\ref{app:class_specific_rates}, but implementing it and evaluating its performance remains a subsequent step.
It would also be interesting to relax the assumption of fixed community memberships and to combine the birth-death mechanism with models that permit community switching, overlapping or mixed-membership structures, as well as the possibility for new communities to emerge over time. Another important avenue is to relax the assumption that the birth-death process is fully observed, and instead consider settings where only discrete-time measurements of the population (e.g.\ presence/absence or counts at $t_0,\dots,t_T$) are available, so that the demographic process becomes partially latent and must be inferred jointly with the community structure. A further promising direction is to let edges evolve in continuous time according to a Markovian dynamics, for instance via a continuous-time Markov chain on edge states; in such settings, the VEM algorithm developed here would no longer yield closed-form updates and it would be natural to resort to MCMC-based or hybrid variational–sampling schemes to approximate the posterior distribution. Finally, applying the BD SBM to additional real datasets in ecology, epidemiology or more finely resolved academic collaboration networks, in which one can follow, over long time periods, the evolution of emerging or trending research topics, would further assess its practical relevance and may suggest additional refinements of the model.


\bmhead{Supplementary information}

\begin{center}
\label{tab:arxiv_math_categories}
\begin{tabular}{ll}
\toprule
Category & Description \\
\midrule
\texttt{math.AC} & Commutative Algebra. \\
\texttt{math.AG} & Algebraic Geometry.  \\
\texttt{math.AP} & Analysis of PDEs. \\
\texttt{math.AT} & Algebraic Topology. \\
\texttt{math.CA} & Classical Analysis and ODEs. \\
\texttt{math.CO} & Combinatorics. \\
\texttt{math.CT} & Category Theory.  \\
\texttt{math.CV} & Complex Variables. \\
\texttt{math.DG} & Differential Geometry.\\
\texttt{math.DS} & Dynamical Systems. \\
\texttt{math.FA} & Functional Analysis. \\
\texttt{math.GM} & General Mathematics. \\
\texttt{math.GN} & General Topology. \\
\texttt{math.GR} & Group Theory. \\
\texttt{math.GT} & Geometric Topology.\\
\texttt{math.HO} & History and Overview. \\
\texttt{math.IT} & Information Theory. \\
\texttt{math.KT} & K-Theory and Homology. \\
\texttt{math.LO} & Logic. \\
\texttt{math.MG} & Metric Geometry. \\
\texttt{math.MP} & Mathematical Physics.  \\
\texttt{math.NA} & Numerical Analysis.  \\
\texttt{math.NT} & Number Theory. \\
\texttt{math.OA} & Operator Algebras. \\
\texttt{math.OC} & Optimization and Control. \\
\texttt{math.PR} & Probability. \\
\texttt{math.QA} & Quantum Algebra.  \\
\texttt{math.RA} & Rings and Algebras. \\
\texttt{math.RT} & Representation Theory.  \\
\texttt{math.SG} & Symplectic Geometry. \\
\texttt{math.SP} & Spectral Theory. \\
\texttt{math.ST} & Statistics Theory.  \\
\bottomrule
\end{tabular}
\captionsetup{type=table,width=\linewidth}
\captionof{table}{arXiv Mathematics subject categories and their descriptions.}
\end{center}

\appendix

\section{ Proofs}
\label{app:proofs}

In this appendix we collect the proofs of Propositions~\ref{prop:gamma} and \ref{prop:pi_beta}.
We first recall the Lagrangian associated with the constraints introduced in
Section~\ref{constraints}.

\subsection{Lagrangian associated with the constraints}

The constraints on the variational parameters
$(\delta,\gamma,\gamma_{\text{mar}})$ can be written in the form:
\begin{itemize}
    \item[(a)] $\sum_k \delta(i,k) = 1$, for all $i \in V_0$;
    \item[(b)] $\sum_{k,n} \gamma(\ell,k,n,n)\,\gamma_{\text{mar}}(\ell-1,k,n) = K-1$,
                for all $\tau_\ell \in \tau$;
    \item[(c)] $\gamma(\ell,k,n,n+b_\ell) + \gamma(\ell,k,n,n) = 1$, for all $\ell,k,n$;
    \item[(d)] $\sum_{k}\sum_{n=0}^{N(\tau_{\ell-1})}
                \gamma(\ell,k,n,n+b_\ell)\,\gamma_{\text{mar}}(\ell-1,k,n) = 1$,
                for all $\tau_\ell \in \tau$,
\end{itemize}
where $b_\ell \in \{-1,1\}$ indicates whether the event at time $\tau_\ell$ is a birth
or a death.

We associate to these constraints the following Lagrange terms:
\begin{align*}\label{lap_constraints}
\left\{ 
\begin{array}{ll}
 \text{Lag(a)}  = &  \displaystyle
   \sum_i \nu_i \Big( \sum_{k} \delta(i,k) - 1 \Big), \\[0.3cm]
 \text{Lag(b)} = & \displaystyle
   \sum_\ell \sigma_\ell \Big( \sum_{k,n} \gamma(\ell,k,n,n)\,
                               \gamma_{\text{mar}}(\ell-1,k,n) - (K-1) \Big),\\[0.3cm]
 \text{Lag(c)} = & \displaystyle
   \sum_{\ell,k,n} \eta_{\ell k n} \Big( \gamma(\ell,k,n,n+b_\ell)
                                       + \gamma(\ell,k,n,n) - 1 \Big), \\[0.3cm]
 \text{Lag(d)} = & \displaystyle
   \sum_{\ell} \upsilon_\ell \Big( \sum_{k}\sum_{n=0}^{N(\tau_{\ell-1})}
                                   \gamma(\ell,k,n,n+b_\ell)\,
                                   \gamma_{\text{mar}}(\ell-1,k,n) - 1 \Big).
\end{array} \right.
\end{align*}

We then consider the augmented objective function,
\[
\mathcal{J} (\delta, \gamma, \nu , \sigma , \eta , \upsilon ) = \mathcal{L}(q) +  \text{Lag(a)} + \text{Lag(b)} + \text{Lag(c)} + \text{Lag(d)},
\]

where $\mathcal{L}(q)$ is the ELBO function defined in \ref{expectation}.

\subsection{Update of $\delta(i,k)$ for $i \in V_0$}
\label{proof:delta}
We maximize $\mathcal{J}(\delta,\gamma,\nu,\sigma,\eta,\upsilon)$ with respect to $\delta(i,k)$ for $i\in V_0$, under constraint~(a). Following the approximation used in \cite{matias2017statistical,agarwal2025clustering}, we neglect the dependence on $\delta(i,k)$ of certain terms in $\mathcal{L}(q)$. Differentiating the resulting objective with respect to $\delta(i,k)$ gives


\begin{equation*}
    \frac{\partial \mathcal{J} }{\partial \delta(i,k)}
    =  \sum_{j \neq i} \sum_{k'} \sum_{\ell \in \Upsilon_{ij}}
       \delta(j,k') \log \phi\left(e_{ij}(t_\ell),k,k'\right)
       + \log \beta_k - \log \delta(i,k) - 1 + \nu_i.
\end{equation*} 
Setting this derivative equal to zero leads to

\begin{equation*}
    \log \delta(i,k)
    = \sum_{j \neq i} \sum_{k'} \sum_{\ell \in \Upsilon_{ij}}
      \delta(j,k') \log \phi\left(e_{ij}(t_\ell),k,k'\right)
      + \log \beta_k - 1 + \nu_i.
\end{equation*}
Hence,
\begin{equation*}
    \delta(i,k)
    = \beta_k\, e^{\nu_i - 1}
      \prod_{j\neq i} \prod_{k'} \prod_{\ell \in \Upsilon_{ij}}
      \phi\left(e_{ij}(t_\ell),k,k'\right)^{\delta(j,k')},
    \qquad \text{for all } i \in V_0.
\end{equation*}
Finally, the normalisation $\sum_k \delta(i,k) = 1$ for each $i$ determines $e^{\nu_i-1}$
as the appropriate normalizing constant and yields the expression proposed.

\subsection{Proof of Proposition~\ref{prop:gamma}}
\label{proof:gamma}

We now maximize $\mathcal{J} (\delta, \gamma, \nu , \sigma , \eta , \upsilon ) $ with respect to $\gamma(\ell,k,n,n+1)$ and
$\gamma(\ell,k,n,n)$ for birth times $\ell \in T_B$. Recall that, if $i_\ell$ is the
individual born at time $\tau_\ell$, then
\[
\delta(i_\ell,k)
= \sum_{n \ge 1} \gamma(\ell,k,n,n+1)\,\gamma_{\text{mar}}(\ell-1,k,n),
\]
so that
\begin{equation}
\frac{\partial \delta (i_\ell, k)}{\partial \gamma(\ell,k,n,n+1)}
= \gamma_{\text{mar}}(\ell-1,k,n).
\end{equation}

\begin{itemize}
    \item Using the above relation and differentiating $\mathcal{J} (\delta, \gamma, \nu , \sigma , \eta , \upsilon ) $ with respect
to $\gamma(\ell,k,n,n+1)$, we obtain
\begin{equation*}
\begin{split}
    \frac{\partial \mathcal{J} } 
         {\partial \gamma(\ell,k,n,n+1)}
    &= \sum_{j \neq i_\ell} \sum_{k'}\sum_{\ell' \in \Upsilon_{i_\ell j}}
       \gamma_{\text{mar}}(\ell-1,k,n)\,\delta(j,k')\,
       \log \phi\big(e_{ij}(t_{\ell ' }),k,k'\big) 
       \\
    &\quad + \gamma_{\text{mar}}(\ell-1,k,n)\log n
           - \gamma_{\text{mar}}(\ell-1,k,n)\log \gamma(\ell,k,n,n+1)
           - \gamma_{\text{mar}}(\ell-1,k,n) \\
    &\quad + \upsilon_\ell\, \gamma_{\text{mar}}(\ell-1,k,n)
           + \eta_{\ell k n}\, \gamma_{\text{mar}}(\ell-1,k,n).
\end{split}
\end{equation*}

Setting this derivative equal to zero and dividing by $\gamma_{\text{mar}}(\ell-1,k,n)$
leads to
\begin{equation*}
    \log \gamma(\ell,k,n,n+1)
   = \sum_{j \neq i_\ell} \sum_{k'}\sum_{\ell' \in \Upsilon_{i_\ell j}}
       \delta(j,k') \log \phi\big(e_{ij}(t_{\ell ' }),k,k'\big)
       + \log n - 1 + \upsilon_\ell + \eta_{\ell k n}.
\end{equation*}

Equivalently,

\begin{equation*}
    \gamma(\ell,k,n,n+1)
    = e^{\upsilon_\ell + \eta_{\ell k n} -1}\, n
      \prod_{j \neq i_\ell} \prod_{k'} \prod_{\ell' \in \Upsilon_{i_\ell j}}
      \phi\big( e_{ij}(t_{\ell ' }),k,k'\big)^{\delta(j,k')},
    \qquad \text{for } \ell \in T_B.
\end{equation*}

\item Similarly, differentiating with respect to $\gamma(\ell,k,n,n)$ gives
\begin{equation*}
\begin{split}
    \frac{\partial \mathcal{J}}
         {\partial \gamma(\ell,k,n,n)}
    &= - \gamma_{\text{mar}}(\ell-1,k,n)\log \gamma(\ell,k,n,n)
       - \gamma_{\text{mar}}(\ell-1,k,n) \\
    &\quad + \sigma_\ell\, \gamma_{\text{mar}}(\ell-1,k,n)
       + \eta_{\ell k n}\, \gamma_{\text{mar}}(\ell-1,k,n),
\end{split}
\end{equation*}
so that
\begin{equation*}
    \log \gamma(\ell,k,n,n)
    = -1 + \sigma_\ell + \eta_{\ell k n}.
\end{equation*}

\end{itemize}

For notational convenience, define
\[
p(\ell,k)
= \prod_{j \neq i_\ell} \prod_{k'} \prod_{\ell' \in \Upsilon_{i_\ell j}}
  \phi\big( e_{ij}(t_{\ell ' }),k,k'\big)^{\delta(j,k')}.
\]
Then we can rewrite
\[
\gamma(\ell,k,n,n+1)
= e^{\upsilon_\ell + \eta_{\ell k n} -1}\, n\, p(\ell,k) \qquad \text{ and }
\qquad
\gamma(\ell,k,n,n)
= e^{\sigma_\ell + \eta_{\ell k n} -1}.
\]

Constraint (c) imposes $\gamma(\ell,k,n,n+1) + \gamma(\ell,k,n,n) = 1$, hence
\begin{equation*}
e^{\eta_{\ell k n}}\Big[e^{\upsilon_\ell-1}\,n\,p(\ell,k)
                     + e^{\sigma_\ell-1}\Big] = 1,
\end{equation*}
which gives
\begin{equation*}
e^{\eta_{\ell k n}}
= \frac{1}{e^{\upsilon_\ell-1}\,n\,p(\ell,k) + e^{\sigma_\ell-1}}.
\end{equation*}

Substituting back, we obtain
\begin{equation*}
\gamma(\ell,k,n,n+1)=
\begin{cases}
\dfrac{e^{\upsilon_\ell}\,n\,p(\ell,k)}
      {e^{\sigma_\ell}+n\,e^{\upsilon_\ell}\,p(\ell,k)}
& \text{if } n\in[0,N_{\ell-1}[,\\[1.1em]
1 & \text{if } n=N_{\ell-1},
\end{cases}
\end{equation*}
and
\begin{equation*}
\gamma(\ell,k,n,n)=
\begin{cases}
\dfrac{e^{\sigma_\ell}}
      {e^{\sigma_\ell}+n\,e^{\upsilon_\ell}\,p(\ell,k)}
& \text{if } n\in[0,N_{\ell-1}[,\\[1.1em]
0 & \text{if } n=N_{\ell-1}.
\end{cases}
\end{equation*}

Finally, setting $\rho_\ell = e^{\sigma_\ell - \upsilon_\ell} > 0$, we can rewrite
\[
\gamma(\ell,k,n,n+1)
= \frac{n\,p(\ell,k)}{\rho_\ell + n\,p(\ell,k)} \text{ and }
\gamma(\ell,k,n,n)
= \frac{\rho_\ell}{\rho_\ell + n\,p(\ell,k)},
\qquad 0 \leq n < N_{\ell-1}.
\]
The value of $\rho_\ell$ is determined from constraint (d) (or equivalently from (b)), which requires
\begin{equation*}
    \sum_{k,n} \gamma(\ell,k,n,n+1)\,\gamma_{\text{mar}}(\ell-1,k,n) = 1,
\end{equation*}
that is,
\begin{equation*}
    \sum_{k,n}
    \frac{n\,p(\ell,k)}{\rho_\ell + n\,p(\ell,k)}\,
    \gamma_{\text{mar}}(\ell-1,k,n) = 1.
\end{equation*}

First, we need to prove the existence of this solution. Define
\[
f_\ell(\rho)
= \sum_{k,n} \frac{n\,p(\ell,k)}{\rho + n\,p(\ell,k)}\,
  \gamma_{\text{mar}}(\ell-1,k,n).
\]
For each fixed $(k,n)$ with $n\geq 1$ and $p(\ell,k)>0$, the map
\[
\rho \mapsto \frac{n\,p(\ell,k)}{\rho + n\,p(\ell,k)}\,\gamma_{\text{mar}}(\ell-1,k,n)
\]
is continuous and strictly decreasing on $(0,\infty)$, with 

\begin{align}
    \begin{split}
        \lim_{\rho \to 0^+}  \frac{n\,p(\ell,k)}{\rho +n\,p(\ell,k)}\,& \gamma_{\text{mar}}(\ell-1,k,n) = \gamma_{\text{mar}}(\ell-1,k,n), \quad \quad \\
        \lim_{\rho \to +\infty}  \frac{n\,p(\ell,k)}{\rho + n\,p(\ell,k)}\, &\gamma_{\text{mar}}(\ell-1,k,n) = 0.
    \end{split}
\end{align}

Hence $f_\ell$ is continuous and strictly decreasing on $(0,\infty)$, and
\[
\lim_{\rho \to 0^+}  f_\ell(\rho)
= \sum_{k}\sum_{n\geq 1} \gamma_{\text{mar}}(\ell-1,k,n), 
\]

is the expected number of non-empty communities at time $\tau_{\ell-1}$.
If $N(\tau_{\ell-1}) \geq 1$ there is at least one individual alive, so at
least one community is non-empty with positive probability, and then the number of non-empty communities is almost surely greater or equal than 1. This implies in particular that
$f_\ell(0^+) \geq 1$, while 
\[
\lim_{\rho \to + \infty }  f_\ell(\rho)
= 0.
\]
By continuity and monotonicity, there exists a unique $\rho_\ell > 0$ such that
$f_\ell(\rho_\ell) = 1$. Therefore the equation
\[
\sum_{k,n}
    \frac{n\,p(\ell,k)}{\rho_\ell + n\,p(\ell,k)}\,
    \gamma_{\text{mar}}(\ell-1,k,n) = 1,
\]
admits a unique solution $\rho_\ell>0$. Since this equation has no closed-form solution in general, $\rho_\ell$ must be found
numerically. This concludes the proof of Proposition~\ref{prop:gamma}.

\subsubsection{Approximation of $\gamma$ for death times}\label{proof:gamma_deaths}

We show here that the parametrization adopted in \eqref{gamma_deaths} admits a unique solution.
For each fixed $(\ell,k)$, define
\[
f_{\ell k}(\rho)
= \sum_{n=0}^{N(\tau_{\ell -1})-1}  
  \frac{\rho n}{1 + \rho n}\;
  \gamma_{\mathrm{mar}}(\ell -1, k, n+1),
\qquad \rho \geq 0.
\]
For $n \geq 1$, the map
\[
\rho \mapsto \frac{\rho n}{1 + \rho n}
\]
is continuous and strictly increasing on $[0,\infty)$, with limits $0$ as $\rho \to 0^+$ and $1$ as $\rho \to \infty$, and derivative $n/(1+\rho n)^2 > 0$. 

Then, $f_{\ell k}$ is continuous and strictly increasing on $[0,\infty)$, and
\[
\lim_{\rho \to 0^+}f_{\ell k}(\rho) = 0, 
\qquad
\lim_{\rho \to \infty} f_{\ell k}(\rho)
= \sum_{n=1}^{N(\tau_{\ell-1})} \widehat{\gamma}_{\mathrm{mar}}(\ell -1, k, n)
= 1 - \widehat{\gamma}_{\mathrm{mar}}(\ell -1, k, 0).
\]
In our model, a death in community $k$ cannot occur when the community is empty, so it is natural to have
\[
\delta(i_\ell,k) \leq 1 - \gamma_{\mathrm{mar}}(\ell -1, k, 0).
\]
By continuity and strict monotonicity of $f_{\ell k}$, there then exists a unique $\rho_{\ell k} \geq 0$ such that
\[
f_{\ell k}(\rho_{\ell k}) = \delta(i_\ell,k),
\]
which justifies the result.

\subsection{Proof of Proposition~\ref{prop:pi_beta}}
\label{proof:pi_beta}

For the proofs of the expressions obtained for the model parameters, we maximize the ELBO $\mathcal{L}(q)$ with respect to $\pi_{k_1k_2}$ and $\beta_k$, under the constraint $\sum_{k=1}^K \beta_k = 1$.

Then we maximize $\mathcal{L}(q)$ in function of $\pi$

\[
\frac{\partial\mathcal{L}(q) }{\partial \pi_{k_1k_2}}
= \sum_{i \neq j} \sum_{k_1,k_2}
  \delta(i,k_1)\,\delta(j,k_2)
  \Bigg[
    \Big(\sum_{\ell \in \Upsilon_{ij}} e_{ij}(t_\ell)\Big) \log \pi_{k_1k_2}
    + \Big(|\Upsilon_{ij}| - \sum_{\ell \in \Upsilon_{ij}} e_{ij}(t_\ell)\Big)
      \log (1-\pi_{k_1k_2})
  \Bigg],
\]

and setting this derivative equal to zero we obtain the result.

To maximize $\mathcal{L}(q)$ with respect to $ \beta$ under the constraint
$\sum_k \beta_k = 1$, we introduce a Lagrange multiplier $\xi$ and consider
\[
\mathcal{J}(\beta,\xi)
= \mathcal{L}(q) + \xi\Big(\sum_{k=1}^K \beta_k - 1\Big)
= \sum_{i \in V_0} \sum_{k=1}^K \delta(i,k)\,\log \beta_k
  + \xi\Big(\sum_{k=1}^K \beta_k - 1\Big).
\]

We differentiate with respect to $\beta_k$:
\[
\frac{\partial \mathcal{J}}{\partial \beta_k}
= \frac{1}{\beta_k} \sum_{i \in V_0} \delta(i,k) + \xi.
\]

Setting this equal to zero yields,
\[
\beta_k = -\frac{1}{\xi} \sum_{i \in V_0} \delta(i,k).
\]
Using the constraint $\sum_k \beta_k = 1$, we obtain the result.

\section{Community-specific birth and death rates in the BD--SBM}
\label{app:class_specific_rates}

In the main part of this work, we assumed shared demographic rates across communities, that is \(\lambda_k \equiv \lambda\), \(\mu_k \equiv \mu\) for all $k\in \{1,\ldots,K\}$. This appendix relaxes this assumption and allows each community to have its own birth and death rates, leading to the following vector parameters
\[
\lambda = (\lambda_1,\dots,\lambda_K), \qquad
\mu = (\mu_1,\dots,\mu_K),
\qquad \text{with } \lambda_k\ge 0,\ \mu_k\ge 0,
\]
where $\lambda_k$ and $\mu_k$ correspond respectively to the birth and death parameters of community $k$. For simplicity, we denote this vector parameters by $\lambda$ and $\mu$ throughout this appendix.
All other components of the BD--SBM (SBM edge model, variational family, and optimization scheme) remain unchanged; only the birth--death contribution to the likelihood/ELBO and the corresponding VM-step updates are modified.

\subsection{Birth--death likelihood with \texorpdfstring{$(\lambda_k,\mu_k)$}{(lambda_k,mu_k)}}
\label{app:bd_likelihood_class_specific}

Let \(Z_i\) be the (latent) community labels of individual \(i\).
Conditionally on the labels, each alive individual of class \(k\) gives birth at rate
\(\lambda_k\) and dies at rate \(\mu_k\).
Let \(N_k(t)\) be the size of community \(k\) at time \(t\), and denote the event times by
\(\tau_1<\cdots<\tau_{M}\), with inter-event durations
\(\Delta_\ell=\tau_\ell-\tau_{\ell-1}\) (and \(\tau_0=t_0\)).
Recall that  \(T_B\) (resp. \(T_D\)) denotes the set of birth (resp. death) event indices.
As in the main text, let \(i_\ell\) be the newborn at \(\tau_\ell\) if \(\tau_\ell\in T_B\),
and the dying individual at \(\tau_\ell\) if \(\tau_\ell\in T_D\).

The total event rate at \(\tau_{\ell - 1}\) is

\begin{equation*}
R(\tau_{\ell-1})
= \sum_{k=1}^K (\lambda_k+\mu_k)\,N_k(\tau_{\ell-1}).
\label{eq:R_class_specific}
\end{equation*}

Under this parametrization, the second term in the likelihood \eqref{Likelihood} can be written as

\begin{align}
 \Pro_{\theta}\big({Z}, \tau, b(\tau) \mid Z^{t_0} \big)
 &= \prod_{\ell =1}^{M} 
    \Bigg[
      e^{  -R(\tau_{\ell -1})\Delta_{\ell} }
      \prod_{k=1}^{K} 
      \Big( \lambda_k  \sum_{i \in V(\tau_{\ell -1})}  Z_{ik} \Big)^{Z_{j_\ell k}\mathds{1} \left\{ b_\ell = 1  \right\} }
      \mu_k^{Z_{j_\ell k}\mathds{1} \left\{ b_\ell = -1 \right\} } 
    \Bigg]  \nonumber \\
 &= \prod_{\ell =1}^{M} 
    \Bigg[
      e^{  -R(\tau_{\ell -1})\Delta_{\ell} }
      \prod_{k=1}^{K} \prod_{n=1}^{N(\tau_{\ell -1})}
      \Big( \lambda_k  n \Big)^{L_{k,n+1}^{\ell } L_{k,n}^{\ell -1}}
      \mu_k^{L_{k,n-1}^{\ell } L_{k,n}^{\ell -1}}  
    \Bigg].  \nonumber 
\end{align}

Taking logs yields the log-likelihood component
\begin{equation*}
\log \Pro_{\theta}\big({Z}, \tau, b(\tau) \mid Z^{t_0} \big)
=
\sum_{\ell=1}^{M+1} \Bigg[
-R(\tau_{\ell-1})\Delta_\ell +
\sum_{k = 1}^K \sum_{n = 1}^{N(\tau_{\ell-1})} L_{k,n+1}^{\ell } L_{k,n}^{\ell -1} \log\!\Big(\lambda_k n \Big)
+
L_{k,n-1}^{\ell } L_{k,n}^{\ell -1} \log \mu_k \Bigg].
\label{eq:bd_loglik_class_specific}
\end{equation*}

\subsection{Contribution to the ELBO}
\label{app:bd_elbo_class_specific}

Recall that \(\gamma_{\mathrm{mar}}(\ell,k,n)\) denotes the variational marginal pmf of \(N_k(\tau_\ell)\), and \(\gamma(\ell,k,n,n')\) the variational transition probability of passing in community \(k\) from $n$ to $n'$ at a time event \(\tau_\ell \in \tau \) (see Section~\ref{subsub:Var_dist}). Under the variational distribution \(q\), the expected total rate satisfies
\[\mathbb{E}_q\left[ R(\tau_{\ell-1}) \right] = \sum_{k=1}^{K} (\lambda_k+\mu_k) \sum_{n=1}^{N(\tau_{\ell-1})} n\,\gamma_{\mathrm{mar}}(\ell -1,k,n).\]

Therefore, the BD contribution to the ELBO, that is, the sum of the second, third, fourth  and fifth terms in  \eqref{expectation} becomes,

\begin{equation*}
\begin{split}
\sum_{\tau_{\ell} \in \tau} \sum_{k=1}^{K} \sum_{n=1}^{N(\tau_{\ell-1})}  & \Bigg[
 -(\lambda_k + \mu_k)\Delta_\ell \, n\, \gamma_{mar}(\ell-1,k,n) \\
 & \quad  +  \gamma(\ell,k,n,n+1)\,\gamma_{mar}(\ell-1,k,n)\log(\lambda_k n) 
   +  \gamma(\ell,k,n,n-1)\,\gamma_{mar}(\ell-1,k,n) \log(\mu_k ) \Bigg] \\
  = &  - \sum_{\tau_{\ell} \in \tau} \sum_{k=1}^{K} \sum_{n=1}^{N(\tau_{\ell-1})}
\left( \lambda_k + \mu_k \right) \Delta_\ell \, n\, \gamma_{mar}(\ell-1,k,n) + \sum_{\tau_{\ell} \in T_D} \sum_{k=1}^{K} \delta(i_\ell , k) \log \mu_k \\
 & \quad + \sum_{\tau_{\ell} \in T_B} \sum_{k=1}^{K}   \sum_{n=1}^{N(\tau_{\ell-1})} \gamma(\ell,k,n,n+1)\,\gamma_{mar}(\ell-1,k,n)\log(\lambda_k n).
\end{split}
\end{equation*}

Then, the ELBO becomes, 

\begin{equation*}
    \begin{split}
   \mathcal{L}(q) & = 
   \sum_{i<j} \sum_{k_1, k_2} \delta(i,k_1) \delta(i,k_2) \sum_{\ell \in \Upsilon_{ij}} \log \phi (e_{ij}(t_\ell ), k_1, k_2) - \sum_{\tau_{\ell} \in \tau} \sum_{k=1}^{K} \sum_{n=1}^{N(\tau_{\ell-1})} \left( \lambda_k + \mu_k \right) \Delta_\ell \, n\, \gamma_{mar}(\ell-1,k,n)\\
    & \quad + \sum_{\tau_{\ell} \in T_D} \sum_{k=1}^{K} \delta(i_\ell , k) \log \mu_k  + \sum_{\tau_{\ell} \in T_B} \sum_{k=1}^{K} \sum_{n=1}^{N(\tau_{\ell-1})}  \gamma(\ell,k,n,n+1)\,\gamma_{mar}(\ell-1,k,n)\log(\lambda_k n) \\
    & \quad +  \sum_{i \in V_{0}} \sum_{k=1}^{K} \delta(i,k)  \log \beta_k - \sum_{i\in V_0} \sum_{k=1}^{K} \delta(i,k) \log \delta(i,k) \\ 
    & \quad +  \sum_{\tau_\ell \in T_B} \sum_{k=1}^{K} \sum_{n=1}^{N(\tau_{\ell-1})} \Big[ \gamma(\ell,k, n, n) \gamma_{\text{mar}}(\ell -1,k, n)\log \gamma(\ell,k, n, n) \\
     & \quad \quad  \quad + \gamma(\ell,k, n, n+1) \gamma_{\text{mar}}(\ell -1,k, n)\log \gamma(\ell,k, n, n+1)\Big].
    \end{split}
\end{equation*}

\subsection{VM-step updates for \texorpdfstring{$\lambda_k$}{lambda_k} and \texorpdfstring{$\mu_k$}{mu_k}}
\label{app:vm_updates_lambda_mu}

We now summarize the coordinate updates used in the VEM algorithm.

In analogy to equation \eqref{eq:delta-fixed-point}, we choose to update $\delta(i,k)$ for $i\in V_{0}$ by solving the fixed-point equation
\[\widehat{\delta}(i,k ) \propto \beta_k \mu_k^{\mathds{1}\{ \tau_i^d \in T_D \} } \prod_{j \neq i } \prod_{k'} \prod_{\ell ' \in \Upsilon_{ij}} \left[ \phi \left( e_{ij}(t_{\ell'} ), k , k' \right) \right]^{\widehat{\delta}(j, k')},   \]
where $\phi \left( e_{ij}(t_{\ell ' }), k , k' \right)$ and $\Upsilon_{ij}$ are given in Proposition  \ref{complete-log-lik}.

\begin{proposition} \label{prop:gamma_app}
For $\ell \in T_B$, we refer to $i_\ell$ as the newborn at time $\tau_\ell$, and we denote by $\widehat{\gamma}(\ell,k,n,n)$ the quantity that maximizes $\mathcal{L}(q)$ with respect to $\gamma(\ell,k,n,n)$ when $\theta = (\lambda,\mu,\pi,\beta)$ is fixed. Then there exists a unique $\rho_\ell > 0$ such that
\[
\widehat{\gamma}(\ell,k,n,n)
= \frac{\rho_\ell}{
    \rho_\ell
    +  \mu_k^{\mathds{1}\{ \tau_{i_\ell}^d \in T_D \} } \lambda_k n \displaystyle\prod_{j \neq i_\ell}
                       \prod_{k'}
                       \prod_{\ell' \in \Upsilon_{i_\ell j}}
                       \Big[ \phi\big( e_{ij}(t_{\ell ' }), k, k' \big) \Big]^{\widehat{\delta}(j,k')}
  },
\qquad \text{for } n \in [0, N_{\ell-1}[,
\]
where the parameter $\rho_\ell$ is determined by the constraint (b) in~\eqref{constraints}.
\end{proposition}

\begin{proof}
The result follows by maximizing $\mathcal{L}(q)$ with respect to $\gamma(\ell,k,n,n)$,
using Lagrange multipliers for the constraints in \eqref{constraints}, in 
analogy to the derivation provided in Appendix~\ref{app:proofs}.
\end{proof}

\begin{remark}
In order to compute $\widehat{\gamma}(\ell,k,n,n)$ in practice, one must first determine
$\rho_\ell$. Since the defining equation for $\rho_\ell$ has no closed-form solution in
general, $\rho_\ell$ is obtained numerically.
\end{remark}

\begin{proposition}\label{prop:lambda_mu}
Let $q \in \mathcal{Q}$ be fixed. Then the maximizers of $\mathcal{L}(q)$ 
with respect to $\lambda_{k}$ and $\mu_k$ are given by
\[
\begin{aligned}
\widehat{\lambda}_k
&= \frac{\sum_{\ell \in T_B}\delta(i_\ell,k)}
        {\sum_{\tau_\ell \in \tau}\ \sum_{n} \Delta_\ell\, n\, \gamma_{\mathrm{mar}}(\ell-1,k,n)},
\end{aligned}
\quad
\begin{aligned}
\widehat{\mu}_k
&= \frac{\sum_{\ell \in T_D} \delta(i_\ell,k)}
        {\sum_{\tau_\ell \in \tau}\ \sum_{n} \Delta_\ell\, n\, \gamma_{\mathrm{mar}}(\ell-1,k,n)}.
\end{aligned}
\]

The updates of $\pi$ and $\beta$ are unchanged and are given in Proposition \ref{prop:pi_beta}.
\end{proposition}

\begin{proof} The result is straightforward by solving the maximization of the  ELBO function $\mathcal{L}(q)$ with respect to the parameters.
\end{proof}

\bibliography{sn-bibliography}

\end{document}